\documentclass[a4paper,UKenglish,11pt]{article}

\usepackage{todonotes}
\usepackage{amsthm}
\usepackage{amsmath}
\usepackage{amssymb}
\usepackage{siunitx}
\usepackage{fullpage}

\usepackage[numbers,sort&compress]{natbib}
\usepackage[utf8]{inputenc} 
\usepackage[T1]{fontenc}    
\usepackage[linkcolor=blue,citecolor=red]{hyperref}       
\usepackage{url}            
\usepackage{booktabs}       
\usepackage{amsfonts}       
\usepackage{nicefrac}       
\usepackage{microtype}      
\usepackage{xcolor}         
\usepackage{graphicx} 
\usepackage{mathtools}
\usepackage{caption}
\usepackage{subcaption}
\usepackage{footnote}
\usepackage{thm-restate}
\usepackage[boxruled,linesnumbered]{algorithm2e}
\usepackage{algpseudocode}
\usepackage{lineno}

\newtheorem{definition}{Definition}
\newtheorem{lemma}{Lemma}
\newtheorem{theorem}{Theorem}
\newtheorem{observation}{Observation}

\usepackage{cleveref}
\Crefname{step}{Step}{Steps}
\Crefname{observation}{Observation}{Observations}

\makeatother

\DeclarePairedDelimiter\paren{(}{)}

\newcommand{\APX}{\ensuremath{\mathsf{APX}}}

\newcommand{\NP}{\ensuremath{\mathsf{NP}}}
\newcommand{\PP}{\ensuremath{\mathsf{P}}}

\newcommand{\tset}{\ensuremath{\mathcal{T}}}
\newcommand{\fset}{\ensuremath{\mathcal{F}}}

\newcommand{\that}{\Hat{T}}
\newcommand{\ttil}{\tilde{T}}
\newcommand{\lptc}{$\ell_p$-{\sc Tree Cover}}
\newcommand{\litc}{$\ell_\infty$-{\sc Tree Cover}}

\newcommand{\antc}{{\sc All-Norm Tree Cover Problem}}
\newcommand{\antcwd}{{\sc All-Norm Tree Cover Problem with Depots}}
\newcommand{\maxsat}{{\sc 3R\{2,3\}L2}}

\newcommand{\reals}{\ensuremath{\mathbb{R}}}

\usepackage{acronym}
\acrodef{VRP}[{\textsf{VRP}}]{Vehicle Routing Problem}
\acrodef{TSP}[\textsf{TSP}]{Traveling Salesperson Problem}
\acrodef{MTSP}[\textsf{MTSP}]{Multi-TSP}
\acrodef{MDMTSP}[\textsf{MDMTSP}]{Multi-Depot Multi-TSP}
\acrodef{CSF}[\textsf{CSF}]{constrained spanning forest}
\acrodef{FPT}[\textsf{FPT}]{fixed parameter tractable}
\acrodef{CPP}[\textsf{CPP}]{Chinese Postperson Problem}
\acrodef{RPP}[\textsf{RPP}]{Rural Postperson Problem}
\acrodef{DRPP}[\textsf{DRPP}]{Depot Rural Postperson Problem}
\acrodef{LP}[\textsf{LP}]{linear program}

\newcommand\restr[2]{{
		\left.\kern-\nulldelimiterspace 
		#1 
		\vphantom{\big|} 
		\right|_{#2} 
}}

\title{Approximate Minimum Tree Cover\\ in All Symmetric Monotone
Norms Simultaneously} 
\date{}

\author{Matthias Kaul\thanks{Universit{\"a}t Bonn, Bonn, Germany. \texttt{mkaul@uni-bonn.de}. Most work done while affiliated with Hamburg University of Technology, Institute for Algorithms and Complexity, Hamburg, Germany.}
\and Kelin Luo\thanks{University at Buffalo, New York, USA. \texttt{kelinluo@buffalo.edu}. Funded by the Deutsche Forschungsgemeinschaft (DFG, German Research Foundation) –  390685813.}
\and Matthias Mnich\thanks{Hamburg University of Technology, Institute for Algorithms and Complexity, Hamburg, Germany. \texttt{matthias.mnich@tuhh.de}}
\and Heiko R{\"o}glin\thanks{Universit{\"a}t Bonn, Bonn, Germany. \texttt{roeglin@cs.uni-bonn.de}. Funded by the Deutsche Forschungsgemeinschaft (DFG, German Research Foundation) – 459420781 and by the Lamarr Institute for Machine Learning and Artificial Intelligence \url{lamarr-institute.org}.}}

\begin{document}
\maketitle

\begin{abstract}
  We study the problem of partitioning a set of $n$ objects in a metric space into $k$ clusters $V_1,\hdots,V_k$. 
  The quality of the clustering is measured by considering the vector of cluster costs and then minimizing some monotone symmetric norm of that vector (in particular, this includes the $\ell_p$-norms).
  For the costs of the clusters we take the weight of a minimum-weight spanning tree on the objects in~$V_i$, which may serve as a proxy for the cost of traversing all objects in the cluster, for example in the context of Multirobot Coverage as studied by Zheng, Koenig, Kempe, Jain (IROS 2005), but also as a shape-invariant measure of cluster density similar to Single-Linkage Clustering.

  This setting has been studied by Even, Garg, K{\"o}nemann, Ravi, Sinha (\emph{Oper. Res. Lett.}, 2004) for the setting of minimizing the weight of the largest cluster (i.e., using $\ell_\infty$) as {\sc Min-Max Tree Cover}, for which they gave a constant-factor approximation.
  We provide a careful adaptation of their algorithm to compute solutions which are approximately optimal with respect to \emph{all} monotone symmetric norms \emph{simultaneously}, and show how to find them in polynomial time.
  In fact, our algorithm is purely combinatorial and can process metric spaces with 10,000 points in less than a second.

  As an extension, we also consider the case where instead of a target number of clusters we are provided with a set of \emph{depots} in the space such that every cluster should contain at least one such depot.
  One can consider these as 
  the fixed starting points of some agents that will traverse all points of a cluster.
  For this setting also we are able to give a polynomial time algorithm computing a constant factor approximation with respect to all monotone symmetric norms simultaneously.
    
  To show that the algorithmic results are tight up to the precise constant of approximation attainable, we also prove that such clustering problems are already  {\sf APX}-hard when considering only one single $\ell_p$ norm for the objective.

\end{abstract}

\thispagestyle{empty}

\newpage
\setcounter{page}{1}

\section{Introduction}
\label{sec:introduction}
A typical clustering problem takes as input a set of $n$ objects in a metric space $(V, d)$, and seeks a partition of these objects into $k$ \emph{clusters} $V_1,\hdots,V_k$, so as to optimize some objective function. 
For example, we might try to place $k$ facilities onto the nodes of an edge-weighted graph, and then assign each remaining node to some facility.
To model the cost of serving these nodes from their facility, one can use the cost of a minimum-weight spanning tree $T_i$ on the subgraph induced by them.
For $i = 1,\hdots,k$, let $w(T_i) = \sum_{e\in E(T_i)}w(e)$ be the weight of tree~$T_i$.  
Historically, these kinds of problems have been studied for two different objectives.
On the one hand, in the {\sc Min-Sum Tree Cover} problem one might want to find $k$ trees $T_1,\hdots, T_k$ to cover all points in~$V$, while minimizing the total length of the service network, i.e., the sum $\sum_{i\in[k]} w(T_i)$ of the weights of the spanning trees.
On the other hand, it is desirable that each facility does not serve too large of a network, so we can instead minimize the weight $\max_{i\in[k]} w(T_i)$ of the heaviest tree, which leads to the {\sc Min-Max Tree Cover} problem~\cite{even2004min, khani2014improved}.

Many fundamental clustering problems were studied under this min-sum objective and min-max objective. 
These objectives can equivalently be considered as that of minimizing the 1-norm, or the $\infty$-norm, of the clustering. 
Examples include
the $k$-median problem and the minimum-load $k$-facility location problem.
These two problems are siblings in that they both deal with the assignment of points (or clients) to $k$ centers (or facilities), with the costs of connecting those points to respective centers.     
Each point $v$ is then assigned to one of these open centers, denoted as $c(v)$.  
The cost $\text{cost}_c$ for each center $c$, which also reflects the cost associated with each facility, is calculated as the sum of distances between the center and all the points allocated to it, i.e, $\text{cost}_c=\sum_{v\in V: c(v)=c} d(v, c)$.
In the $k$-median problem~\cite{guha1999greedy,byrka2017improved, li20131, ahmadian2019better}, the objective is to minimize the $\ell_1$-norm of $\{\text{cost}_c\}_{c}$, i.e.,  $\sum_{c} \text{cost}_c$, which represents the total cost of all clusters; whereas the minimum-load $k$-facility location problem~\cite{ahmadian2018approximation} seeks to minimize the maximum load of an open facility, symbolised by the $\ell_{\infty}$-norm function of $\{\text{cost}_c\}_{c}$, i.e., $\max_c \text{cost}_c$.

There is a diverse array of cost functions that could be applicable to a variety of problem domains, and often, efficient algorithms are crafted to suit each specific objective. 
However, it is crucial to note that an optimal solution for one objective may not perform well for another.
For instance, a solution for an instance of the {\sc Tree Cover} that minimizes the $\ell_1$-norm might be particularly inefficient when it comes to minimizing the $\ell_\infty$-norm, and vice versa (see examples in \Cref{fig:CounterexamplesOneNormGood} and~\Cref{fig:CounterexamplesInftyNormGood}).
Therefore, one may wonder what the ``generally optimal'' solution would be for a given problem.
Finding such a generally optimal solution is the task of the following problem:

\begin{definition}[All-norm clustering problem]
\label{def:allnormapproxratio}
  An \emph{all-norm clustering problem} takes as input a metric space $(V, d)$, a cost function $w: \mathcal{P}(V)\rightarrow \mathbb{R}_{\ge 0}$, and a non-negative integer $k$.
  The goal is to partition~$V$ into clusters $V_1,\hdots,V_k$ which minimize
  \begin{equation*}
    \alpha= \max_{p\in \mathbb R_{\geq 1} \cup \{\infty\}} \frac{ \left(\sum_{i} \left(w(V_i)\right)^p\right)^{1/p}}{OPT_p},
  \end{equation*}
  where $OPT_p$ denotes the value of an optimal solution for the $k$-clustering problem under the $\ell_p$-norm objective.
  Here, $\alpha$ is referred to as the \emph{all-norm approximation factor}.
\end{definition}

Our focus in this paper is the all-norm clustering problem where the cost $w(V_i)$ of each cluster is the cost of a minimum spanning tree on $V_i$ with respect to the metric $d$.
We call this problem the {\sc All-Norm Tree Cover} problem in line with the naming convention in the literature, and denote its instances by $(V,d,w, k)$.
(Note that, for simplicity, we only consider $\ell_p$-norms here.
However, the proofs turn out to work for any norm which is monotone and symmetric, cf. Ibrahimpur and Swamy~\cite{ibrahimpur2020approximation} for a detailed introduction to such norms.)
Observe that the number~$k$ of clusters is part of the input, i.e., it is \emph{not} fixed.

The choice of the cost of a minimum spanning tree might appear to be somewhat arbitrary here, but spanning trees are the key connectivity primitive in network design problems. 
Any constant-factor approximation for {\sc All-Norm Tree Cover} will, for example, transfer to a constant-factor approximation for the all-norm version of the {\sc Multiple Traveling Salesperson Problem}~\cite{Bekts06Mtsp} where we ask to cover a metric space by $k$ cycles instead of trees.
This use of spanning trees as a proxy for the traversal times of the clusters was a key ingredient for by Zheng et al. \cite{zheng2007Robot} to partition a floorplan into similar-size areas to be served by different robots.

Let us observe that solving the {\sc All-Norm Tree Cover} problem is non-trivial, as a solution which is good for one objective (i.e., for one particular norm $\ell_p$) may be bad for another objective (i.e., for another norm $\ell_{p'}$), and vice-versa.
For examples of this phenomenon, see \Cref{fig:CounterexamplesOneNormInftyNorm}.
\begin{figure}[ht]
  \centering
  \begin{subfigure}[b]{0.48\textwidth}
    \centering
    \includegraphics[width=\textwidth]{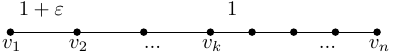}
    \caption{The path forms a metric where any two neighboring nodes in the set $\{v_1,\hdots,v_k\}$ have distance $1+\varepsilon$, while the distance of all other edges is $1$.
    In this case, the optimal solution under the $\ell_1$-norm involves removing all edges with distance $1+\varepsilon$, resulting in an $\ell_\infty$-norm objective cost $n-k$.
    However, the optimal solution cost for the $\ell_\infty$-norm is less than~${n(1+\varepsilon)}/{k}$.\label{fig:CounterexamplesOneNormGood}}
  \end{subfigure}
  \hfill
  \begin{subfigure}[b]{0.48\textwidth}
    \centering
    \includegraphics[width=\textwidth]{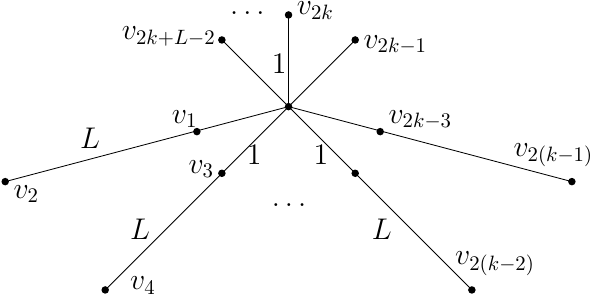}
    \caption{The spider graph includes $2k+L-1$ nodes. All edges connecting a node $v_{2i-3}$ (where $i\in[k]$) and~$v_i$ (where $i=2k-1,\hdots,2k+L-2$) to the center have distance $1$, and all edges $\{v_{2i-1}, v_{2i}\}_{i\in [k-1]}$ have distance a large integer $L$. 
    In this case, the optimal solution under the $\ell_\infty$-norm involves removing all edges connected to the nodes $\{v_{2i-1}\}_{i\in[k-1]}$, resulting in an $\ell_1$-norm objective cost of $kL$. However, the optimal solution for the $\ell_1$ norm has cost~$k+L-1$.\label{fig:CounterexamplesInftyNormGood}}
  \end{subfigure}
  \caption{Two instances of the {\sc All-Norm Tree Cover} problem. \Cref{fig:CounterexamplesOneNormGood} is an instance where an optimal $\ell_1$-norm solution does not return a good approximation in the $\ell_\infty$-norm; \Cref{fig:CounterexamplesInftyNormGood} is an instance where an optimal $\ell_\infty$-norm solution does not return a good approximation in the $\ell_1$-norm.\label{fig:CounterexamplesOneNormInftyNorm}}
\end{figure}

The {\sc All-Norm Tree Cover} problem, alongside the related path cover and cycle cover problems, involves covering a specified set of nodes of a(n edge-weighted) graph with a limited number of subgraphs. 
These problems have attracted significant attention from the operations research and computer science communities due to their practical relevance in fields like logistics, network design, and data clustering.
For instance, these problems are naturally applicable in scenarios such as vehicle routing, where the task involves designing optimal routes to service a set of customers with a finite number of vehicles. 
Different optimization objectives could be considered depending on the specific requirements. 
One could aim to minimize the maximum waiting time for any customer, an objective that is equivalent to the $\ell_\infty$-norm of the cost function associated with each vehicle. 
This ensures fairness, as it attempts to prevent any single customer from waiting excessively long.
Alternatively, one could aim to minimize the total travel time or cost, which corresponds to the $\ell_1$-norm of the cost function across all vehicles~\cite{bellmore1974transformation, bazgan2005approximation}.
This objective seeks overall efficiency, making it beneficial from an operational perspective as it reduces fuel consumption and allows for more customers to be serviced within the same time frame~\cite{frederickson1976approximation, bazgan2005approximation, even2004min, khani2014improved}. 
Understanding these problems in an all-norm setting thus enables the development of routing strategies that are adaptable to different priorities, including customer satisfaction, operational efficiency, or a balance between the two.  
Given that minimum spanning trees provide constant-factor approximations to traveling salesperson tours~\cite{christofides1976worst}, we explore the possibility of covering the nodes of a graph with $k$ trees.

In a scenario where each vehicle already has an assigned station, and the task is limited to the assignment of nodes to the vehicles, we are confronted with a variation of the {\sc All-Norm Tree Cover} problem known as the {\sc All-Norm Tree Cover with Depots} problem (also commonly referred to as the {\sc Rooted All-Norm Tree Cover} problem), denoted by $(V, d, w, D)$ where $D\subseteq V$ is a set of depots~\cite{even2004min}. 
These problems under all norms will be formally defined and addressed in the subsequent sections of this paper. 

\subsection{Our contributions}
Computing optimal tree covers is $\mathsf{NP}$-hard even for the \emph{single} $\ell_\infty$-norm; for this norm, it admits a constant-factor approximation in polynomial time.
Our goal in this paper is thus to design constant-factor approximations for {\sc All-Norm Tree Cover}, that is, for \emph{all} monotone symmetric norms \emph{simultaneously}.
Building upon the example in \Cref{fig:CounterexamplesOneNormInftyNorm}, it becomes clear that a constant-factor approximation for one norm objective does not necessarily guarantee a constant approximation for another norm.
Meanwhile, achieving an all-norm approximation factor better than $3/2$ within polynomial time is infeasible, as a lower bound of $3/2$ on the approximability of the {\sc Tree Cover} problem under the $\ell_\infty$-norm has been demonstrated by Xu and Wen~\cite{xu2010approximation} (assuming $\mathsf{P}\neq \mathsf{NP}$).
In previous work, Even et al.~\cite{even2004min} proposed a $4$-approximation algorithm for the {\sc Min-Max Tree Cover} problem, which is representative of the $\ell_\infty$-norm.
That algorithm was subsequently refined by Khani and Salavatipour~\cite{khani2014improved}, who devised a $3$-approximation.
However, these existing algorithms fail to guarantee a constant approximation factor for optimal solutions under other norms, proving to be unbounded specifically, this holds even for the $\ell_1$-norm (for an example, see \Cref{fig:CounterexampleMinMaxAlgo} in \Cref{sec:TCnodepots}). 

Our first main contribution is a polynomial-time constant-factor approximation algorithm for the {\sc All-Norm Tree Cover} problem. 
Our algorithm amalgamates the strategies of suitable algorithms for both the $\ell_1$-norm and the $\ell_\infty$-norm.
It is well-understood that an optimal solution for the $\ell_1$-norm objective involves successively eliminating the heaviest edges until precisely $k$ components remain. 
Conversely, the algorithm for $\ell_\infty$-norm objective concludes when the number of trees generated by the edge-decomposition with respect to a guessed optimal value~$R$ is about to exceed $k$~\cite{even2004min}.
Our proposed method continually approximates the decomposed trees towards an ideal state.  
So in that sense  the algorithm proposed by Even et al.~\cite{even2004min} is effectively refined to solve other norm problems.
Notably, our algorithm produces a feasible solution that is at most double the optimal solution for the $\ell_1$-norm and four times the optimal solution for the $\ell_\infty$-norm simultaneously.


\begin{restatable}{theorem}{TCnodepots}
\label{thm:TCnodepots} 
  There exists a polynomial-time $O(1)$-approximation algorithm for the {\sc All-Norm Tree Cover} problem. 
\end{restatable}

We next consider the more involved {\sc All-Norm Tree Cover} problem with depots.
For {\sc All-Norm Tree Cover} with depots, each tree in a tree cover solution is rooted at a specific depot.
The special problem case of the single $\ell_{\infty}$-norm is the {\sc Min-Max Tree Cover} problem with depots, for which Even et al.~\cite{even2004min} devised a $4$-approximation algorithm.
Subsequently, Nagamochi~\cite{nagamochi2005approximating} enhanced this framework by offering a $(3 - \frac{2}{k+1})$-approximation algorithm under the restriction that all depots are located at the same node.
However, neither of those algorithms can assure any constant approximation factor for the $\ell_1$-norm objective (refer to the example in \Cref{fig:CounterexampleRootedTreeCoverAlgo}).  
\begin{figure}
  \centering
  \includegraphics[width=0.9\textwidth]{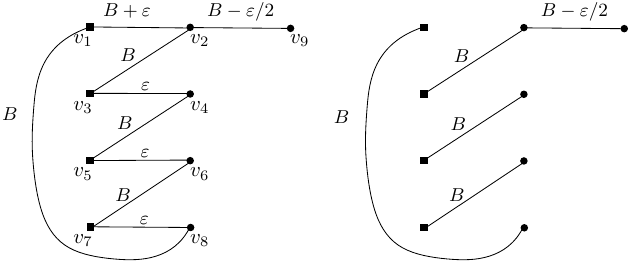}
    \caption{Example instance where the algorithm of Even et al.~\cite{even2004min} does not return a good approximation in the $\ell_1$-objective. The left figure is an example of the graph metric when $k=4$ and the right figure is the final solution obtained by the rooted-tree-cover algorithm by Even et al.~\cite{even2004min}. }
    \label{fig:CounterexampleRootedTreeCoverAlgo}
\end{figure}

Our key contribution extends the $O(1)$-approximation algorithm for {\sc All-Norm Tree Cover} from \autoref{thm:TCnodepots} to the problem with depots:
\begin{restatable}{theorem}{TCdepots}
\label{thm:TCdepots}
  There exists a polynomial-time $O(1)$-approximation algorithm for the {\sc All-Norm Tree Cover with Depots} problem.
\end{restatable}

Our methodology comprises three stages:  
\begin{enumerate}
    \item  Initially, we partition the nodes according to their distances to the depots, where partition class $i$ contains all nodes at distance between $2^{i-1}$ and $2^i$.
    \item We then apply a version of the algorithm of Even et al.~\cite{even2004min} in each class separately to find a good ``pre-clustering'' of the nodes. This is simplified by the previous partitioning since it allows us to assume that all nodes are at the same distance to the depots, up to a factor of~$2$.
    \item Finally, we assign the clusters from the second step to the depots by iteratively computing matchings of clusters to depots, considering ever larger clusters, and allowing them to be matched to ever more distant depots.
    In this way, we essentially maintain a running estimate of the necessary tree weights, updating it when the matching process does not succeed in assigning all nodes to some depot.
If it does succeed, we can show that (up to constant factors) no cluster is larger than the estimate, and that the estimate is correct, i.e., every solution has trees at least as large as predicted by the estimate.
\end{enumerate}

Note that both the depot and non-depot algorithms will require only very simple algorithmic primitives such as minimum spanning trees and bipartite matchings.
Thus, our algorithms can be implemented to run very quickly using purely combinatorial techniques, in particular without resorting to linear programming which can be impractically slow (see also Davies at al. \cite{Davies2023fastCombinatorial} for recent work on combinatorial clustering algorithms). 
In fact, an implementation of our algorithm can process metric spaces with 10,000 points in less than a second, see \Cref{sec:experiments} for details.

Our final contribution complements the algorithmic result from \autoref{thm:TCdepots} by a complexity-theoretic lower bound: we show that one cannot expect polynomial-time approximation schemes to exist, even in the presumably easier setting \lptc{} {\sc with Depots} where we only want to approximate the optimum with respect to \emph{one specific} $\ell_p$-norm:
\begin{restatable}{theorem}{APXhardness}
  For every $p\in (1,\infty]$ there exists a constant $c$ such that \lptc{} {\sc with Depots} is \NP-hard to approximate within a factor $c$. 
  The \NP-hardness holds under randomized reductions.
\end{restatable}
Specifically, we show that \lptc{} is \NP-hard to approximate to a factor 
\begin{equation*}
\left[ \frac{(106-\frac{1}{4} + \frac{\varepsilon}{2}) 3^p + (\frac{1}{8} - \frac{\varepsilon}{4}) (2^p+4^p)}{(106-2\varepsilon) 3^p + \varepsilon  (2^p + 4^p)}\right]^{1/p} > 1,
\end{equation*}
for any choice of $\varepsilon > 0$.

\section{Preliminaries}
\label{sec:prel}
We set up some formal definitions.
We will identify metric spaces and metrically edge-weighted graphs to allow for an easier treatment of connectedness, trees, and similar graph-theoretic objects.
For any such graph~$G$, we always keep in mind the underlying metric space $V$ induced by the shortest-path metric on~$G$.
We also assume that the metrics used are integral, for ease of presentation. 
Our results extend to rational metrics by rescaling the metric.

We here concern ourselves with tree covers of graphs, which are to be defined as follows:
\begin{definition}[Tree cover]
  For a graph $G$ , a \emph{tree cover} is a collection $\{T_1,\dots, T_k\}$ of trees for which $\bigcup_{i}V(T_i) = V$.
  We call $k$ the \emph{size} of such a tree cover.
\end{definition}

\begin{definition}[Tree cover with depots]
  For a graph $G$ and a depot set $D\subseteq V$, a \emph{tree cover with depots} is a tree cover $\{T_1,\dots, T_{|D|}\}$ of $G$, where each $T_i$ contains a unique depot from $D$.    
\end{definition}

Notice in particular that there is no expectation of disjointness for the trees.
If disjointness is required, we may assign every node to exactly one of the trees currently containing it and recompute minimum-weight spanning trees for each of the resulting clusters. 
This increases the cluster weights by at most a factor of $2$ due to the gap between Steiner trees and minimum-weight spanning trees (see \Cref{lem:steinerMST}). 
For any connected subgraph $H$ of a graph $G$, we use $w(H)$ to denote the weight of a minimum-weight spanning tree in the subgraph $H$. 
It is important to note that this weight can differ from the sum of the weight of the edges of $H$.
We then use the $p$-norm $(p \ge 1)$ as a measure of the quality of a tree cover.
Specifically, the cost of a tree cover $\{T_1, \hdots, T_k\}$ is defined as the $p$-norm of the corresponding tree weights, that is, $w_p=\left(\sum_{i\in[k]} \left(w(T_i)\right)^p \right)^{1/p}$.

There are two natural optimization questions for tree covers  (and for tree covers with depots) which have been considered in the literature: The {\sc $\ell_1$-Tree Cover} problem, where the aim is to minimize the sum of the weights of the $k$ trees, and the {\sc $\ell_\infty$-Tree Cover} problem in which the objective is to minimize the weight of the heaviest tree within the cover.
The {\sc $\ell_1$-Tree Cover} problem admits a polynomial-time algorithm, regardless of the presence of depots; in contrast, the {\sc $\ell_\infty$-Tree Cover} problem is $\mathsf{APX}$-hard, again regardless of whether depots are present or not.

We consider the interpolation between these two problems by allowing arbitrary $\ell_p$ norms: 

\begin{definition}[\lptc{} (resp. \lptc{} {\sc with Depots} )]
  Given a graph $G=(V, d)$ and  some integer $k$ (resp. depots $D\subseteq V$) and a $p\in [1,\infty)$, find a tree cover $\{T_1 ,\dots,T_k\}$ (resp. with depots) which minimizes the expression 
    $OPT_{p,k}:=\left(\sum_{i = 1}^k(w(T_i))^p\right)^{1/p}.$
\end{definition}

We will usually omit $k$ if it is clear from context.
All of the $\ell_p$-variants (except $p = 1$) are $\mathsf{NP}$-hard, as the proof of hardness for the $\ell_\infty$-case works for all values of $p > 1$~\cite{xu2010approximation}.  

\begin{definition}[\antc{} (resp. \antcwd{})]
  Given a graph $G=(V, d) $ and some integer $k$ (resp. depots $D\subseteq V$), find a tree cover $\{T_1,\hdots,T_k\}$ (resp. with depots) which minimizes 
    $\max_{p\in[1, \infty)}  \frac{\left(\sum_{i = 1}^k(w(T_i))^p\right)^{1/p}}{OPT_p} \enspace .$
\end{definition}

To distinguish the problem versions for a \emph{fixed} norm $p$ from those for \emph{all $p$ simultaneously}, we denote by $(\ell_p, k)$ (or $(\ell_p, D)$ for problems involving depots) the tree cover problem for a specific value of $p$.
Meanwhile, we use $(\ell_{p\in[1, \infty)}, k)$ (or $(\ell_{p\in[1, \infty)}, D)$ when depots are involved) to represent the {\sc All-Norm Tree Cover} problem that encompasses all values of $p\in[1, \infty)$.
We omit naming the underlying metric space $(V,d)$, unless explicitly stated otherwise.

We will state a result about the relationship between the cost of minimum-weight Steiner trees and minimum-weight spanning trees, as we will harness this argument repeatedly:
\begin{lemma}[Kou et al. \cite{kouFastAlgorithmSteiner1981}]
\label{lem:steinerMST}
  Let $(V,d)$ be a metric space with some terminal set $F\subseteq V$, let $\that$ be a minimum-weight (Steiner) tree containing all nodes from $F$, and let $T$ be a minimum-weight spanning tree of $(F,\restr{d}{F})$.
  Then $w(T) \leq (2-2/|F|)\cdot w(\that)$.
\end{lemma}
\begin{proof}
  An easy way to obtain the result is to take $\that$ and use the tree-doubling heuristic to compute a traveling salesperson tour $\mathcal T$ of $F$ in $(V,d)$, costing at most $2\cdot w(\that)$.
  Removing the heaviest edge of $\mathcal T$ yields a (perhaps not yet minimal) spanning tree $T'$ in $(F,\restr{d}{F})$ of weight at most $(2-2/|F|) w(\that)$, allowing us to conclude that $w(T) \leq w(T') \leq (2-2/|F|) w(\that)$.
\end{proof}
\Cref{lem:steinerMST} will prove useful as we occasionally want to forget about some nodes of our instance.
This may actually increase overall costs, since certain Steiner trees become unavailable. 
However, the lemma ensures that the increase in cost is bounded.

\section{Simultaneous Approximations for the All-Norm Tree Cover Problem}
\label{sec:TCnodepots}

First, to provide some some good intuition of the key techniques for the general case, we show in detail that the {\sc Tree Cover} algorithm of Even et al.~\cite{even2004min} gives a constant-factor approximation to the $\ell_p$-tree cover problem without depots provided that it returns $k$ trees.
Indeed, we show that the argument will work for all monotone symmetric norms.
The original algorithm works as follows, for a given graph $G = (V,d)$:
\begin{enumerate}
  \item Guess an upper bound $R$ on the optimum value of the $\ell_\infty$-norm tree cover problem for instance $G$, where initially $R=1$.
  \item Remove all edges with weight larger than $R$ from $G$, and compute a minimum-weight spanning tree $T_i$ for each connected component of the resulting graph.
  \item Check that $\sum\lfloor \frac{w(T_i) + 2R}{2R} \rfloor = k$; otherwise, reject~$R$ and go to step $1$ with $R := 2R$. \label[step]{alg:item:exactlyK}
  \item Call a spanning tree $T_i$ \emph{small} if $w(T_i) < 2R$.
  \item Call a spanning tree $T_i$ \emph{large} if $w(T_i) \geq 2R$.
    Decompose each such tree into edge-disjoint subtrees $\Tilde{T}_j$ such that $2R \leq w(\Tilde{T}_j) \leq 4R $ for all but one residual $\Tilde{T}_j$ in each component\footnote{
    This can be done with a simple greedy procedure, cutting of subtrees of this weight.
    For the details of this algorithm, we refer to Even et al. \cite{even2004min}.}. \label[step]{alg:item:Decomposition}
  \item Output the family of all small spanning trees, as well as all subtrees $\Tilde{T}_j$ created in \Cref{alg:item:Decomposition}.
\end{enumerate}

Even at al. show that this procedure will output \emph{at most} $k$ trees if $OPT_\infty \leq R$, although they relax the condition in \Cref{alg:item:exactlyK} to $\sum\lfloor \frac{w(T_i) + 2R}{2R} \rfloor  \leq k$.
For technical reasons, however, we need the equality in this spot.
Since every tree has weight most $4R$, the algorithm evidently gives a $4$-approximation in the $\ell_\infty$-norm if the correct value of $R = OPT_{\infty,k}$ is guessed. 
Otherwise, one can use binary search to find, in polynomial time, the smallest $R$ for which one gets at most $k$ trees.

In general, however, we notice that this algorithm will sometimes compute fewer than $k$ trees, causing it to not return a good approximation for the other $\ell_p$-norms, not even for the $\ell_1$-norm (see \Cref{fig:CounterexampleMinMaxAlgo}).
\begin{figure}
  \centering
  \includegraphics[width=0.95\textwidth]{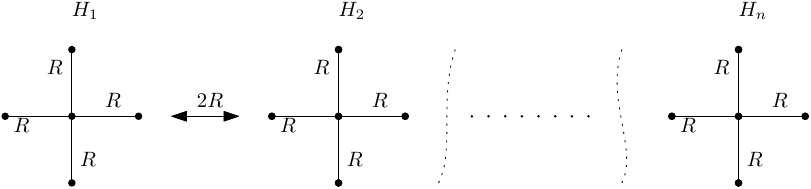}
  \caption{Example instance where the algorithm of Even et al.\ does not return a good approximation in the $\ell_1$-objective. The instance consists of $n$ identical copies $H_i$ of the $4$-star where all edges have length $R$ and the copies are pairwise at distance $2R$. For $k=5n-1$, the instance has $OPT_\infty = OPT_1 = R$, but the algorithm will return $2n$ trees, each of weight $2R$. Observe that the partition requirement in \Cref{alg:item:Decomposition} would also be fulfilled if the trees are not further partitioned and kept at size $4R$, however the algorithm of Even et al. produces the former solution. Both solutions do not achieve a good approximation in the $\ell_1$ objective.\label{fig:CounterexampleMinMaxAlgo}}
\end{figure}

For this reason, we need the strengthened property in \Cref{alg:item:exactlyK}.
Then, however, notice that such a value $R$ does not necessarily exist; for example, for the instance in \Cref{fig:CounterexampleMinMaxAlgo} this is the case.
We will describe later how to avoid this issue of non-existence; for now, we assume that such a value~$R$ exists and can be computed in polynomial time.

So suppose that the algorithm has returned $k$ trees $T_1,\dots,T_k$, where we simply remove some edges should the algorithm return fewer than $k$ trees.
This removal only improves the objective value, so we will assume---without loss of generality---that the algorithm already returned $k$ trees.
As a warm-up, we will start by considering only the case that $p=1$, i.e., we show that this algorithm computes a solution that is a good approximation for $(\ell_1,k)$ tree cover.
\begin{lemma}
\label{lem:l1Approximation}
  Let $\{T_1,\dots,T_k\}$ be the set of trees returned by the algorithm for some fixed value of~$R$.
  Then we have
    $\sum w(T_i) \leq 2 OPT_{1,k} \enspace .$
\end{lemma}
\begin{proof}
  We observe first that an optimal solution for $(\ell_1,k)$ tree cover can be computed by removing iteratively the heaviest edge as long as this does not cause the graph to decompose into more than $k$ components, as this procedure is equivalent to running Kruskal's algorithm.
  As the optimal solution contains no edges of weight greater than $R$, we consider the graph $G' := G - \{e \mid d(e) > R\}$.
  Let~$k_s$ be the number of small spanning trees $S_i$ computed by the algorithm, and let $k_\ell$ be the number of large spanning trees $L_i$. 
  Then we certainly have $\sum w(T_i) \leq \sum w(S_i) + \sum w(L_i)$, since the trees computed by the algorithm are edge-disjoint subtrees of the initial spanning trees.
  Further, we have 
  \begin{equation*}
  OPT_{1,k} \geq \sum w(S_i) + \sum w(L_i) - (k-k_s-k_\ell)R,
  \end{equation*}
  because the optimal solution will remove $k-k_s-k_\ell$ edges from within the spanning trees computed by the algorithm, each of weight at most $R$.
  But notice that we can use \Cref{alg:item:exactlyK} to obtain
  \begin{align*}
    k &= \sum\left\lfloor \frac{w(S_i)}{2R} + 1 \right\rfloor + \sum\left\lfloor \frac{w(L_i)}{2R} + 1 \right\rfloor \leq k_s + k_l + \sum \frac{w(L_i)}{2R},
  \end{align*}
  which implies that $(k - k_s - k_\ell) 2R \leq \sum w(L_i)$.
  This inequality allows us to conclude that
  \begin{align*}
    OPT_{1,k} &\geq \sum w(S_i) + \sum w(L_i) - \frac{1}{2}\sum w(L_i) \geq \frac{1}{2} (\sum w(S_i) + \sum w(L_i)) \geq \frac{1}{2} \sum w(T_i) \enspace .\tag*{\qedhere}
  \end{align*}
\end{proof}

Notice that the strengthened version of \Cref{alg:item:exactlyK} is indeed necessary here.

We can use a similar technique to show that the solution computed by the algorithm is a constant-factor approximation for $(\ell_p,k)$ tree cover for any choice of $p$, and with a constant of approximation independent of $p$. 
The key argument will be to show that an optimal solution to $(\ell_p,k)$ tree cover must, as in \Cref{lem:l1Approximation}, have a total weight comparable to that of the solution computed by the algorithm. 
In a second step, one can then show that the algorithm distributes the total weight of its trees fairly evenly, because the large trees all have weight between $2R$ and~$4R$.

Meanwhile, for the small trees we can demonstrate that choosing them differently from the algorithm will incur some cost of at least $R$.
Thus, either the algorithms' solution has correctly chosen the partition on these small trees, or the optimal solution actually has a heavy tree not in the algorithms' solution, which we can use to pay for one of the large trees that the algorithm has, but the optimal solution does not.
Notice that, if the algorithm achieves an equal distribution of some total weight that is comparable to the total weight of an optimal solution under the $\ell_p$-norm, it also achieves a good approximation of $OPT_{p}$ due to convexity of the $\ell_p$-norms.

To start with, we will only give a rough analysis of the quality of the solution returned by the algorithm, although it already shows a constant factor of approximation.
Thereafter, in \Cref{apx:sec:noDepotImproved}, we provide a more precise analysis of the algorithm to argue that it achieves a smaller constant approximation factor, specifically a $9$-approximation. 

Let us fix some $p \in [1,\infty)$ and any $(\ell_p,k)$ tree cover solution $\{\that_1,\dots,\that_k\}$ (you may imagine that it is optimal). 
Then let $k_s$ be the number of connected components computed by the algorithm that had a small spanning tree, and call those components $S_1,\dots,S_{k_s}$. 
Similarly, for the large spanning trees, let there be $k_\ell$ components $L_1,\dots,L_{k_\ell}$.
Now we denote by $E_\ell$ the set of edges that are both contained in some~$\that_i$ and in the cut induced by some $S_j$ or $L_j$.
In particular, all these edges have weight at least $R$.
We count separately the small~$T_i$ incident to some edge from $E_\ell$, say there are~$k_{s,1}$ of them.
We will also denote by~$k_{s,2}$ the number of $S_i$ for which one of the $\that_j$ is a minimum spanning tree, and denote the set of these $\that_j$ as $T^=$.
Note that any small component not counted by~$k_{s,1}$ or $k_{s,2}$ is split into at least two components by the $\that_i$.
For an illustration of this setting, we refer to \Cref{fig:DepotlessIllustration}.
\begin{figure}
  \centering
  \includegraphics[width=0.95\textwidth]{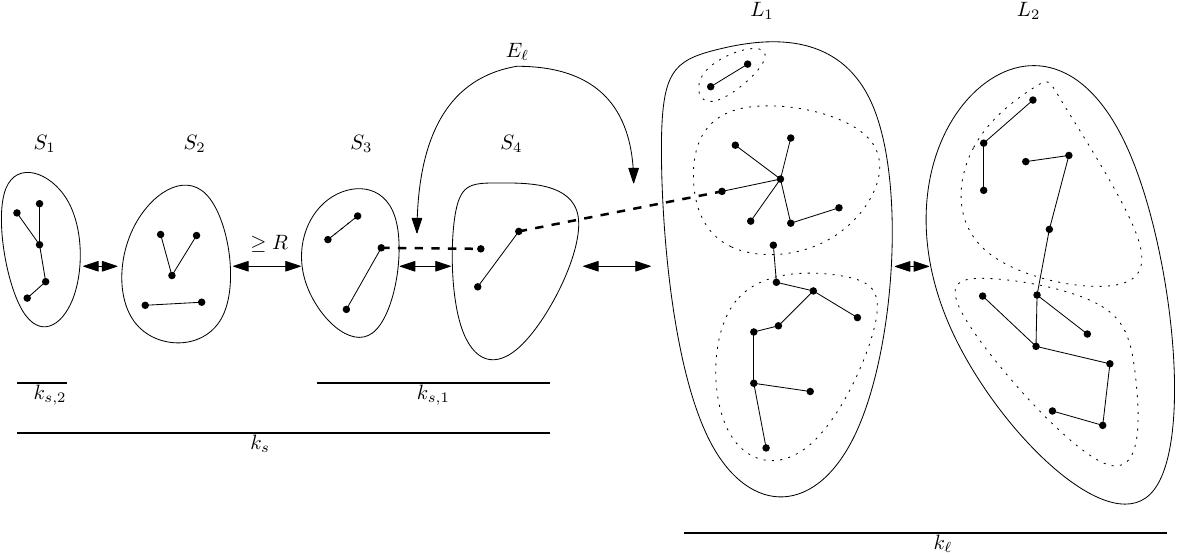}
  \caption{Illustration of how the algorithm's solution and some optimal solution can align against each other. The algorithm's partition of the graph into connected component is indicated by solid zones, the further subdivision of the large components by dotted zones. The trees of the optimal solution are drawn, and all edges crossing the boundary of a connected component are dashed.\label{fig:DepotlessIllustration}}
\end{figure}

We can now start to measure the total sum of edge weights in the $\that_j$, i.e., $\sum_jw(\that_j)$. 
We begin by relating it to the small trees of the algorithm's solution that do not agree with the optimal solution.
\begin{observation}
\label{obs:l1ShortComponents}
  It holds that $\sum_jw(\that_j) - \sum_{\that_i \in T^=} w(\that_i) \geq k_{s,1}\frac{R}{2}$.
\end{observation}
\begin{proof}
  We have $k_{s,1}$ pairwise disjoint node sets in $G$, each incident to at least one edge of weight at least~$R$ present in $E_\ell$, so we get $|E_\ell| \geq k_{s,1}/2$ from the handshake lemma, and thus
  \begin{equation*}
    \sum_jw(\that_j) - \sum_{\that_i \in T^=} w(\that_i) \geq |E_\ell|R \geq \frac{k_{s,1}}{2}R,
  \end{equation*}
  noting that the trees in $T^=$ do not contain any of the edges from $E_\ell$.
\end{proof}
To compare the total weight against the number of large trees, we can use a similar argument as in \Cref{lem:l1Approximation}. 
We again note that the initial spanning trees have weight at least $(k-k_\ell-k_s)2R$, and any tree cover cannot remove too many edges from them. 
The second part of this argument is no longer true though, since it is now possible for a solution to have many edges between the components of $G - \{ e \mid  d(e) \geq R\}$, allowing it to remove many edges within the components
However, a careful analysis will show that this case does not pose a problem, since such inter-component edges are themselves heavy.
\begin{observation}
\label{obs:l1LargeComponents}
  We have $\sum_jw(\that_j) \geq (k  - k_\ell - k_{s,1} - k_{s,2})R + \sum_{\that_i \in T^=} w(\that_i)$.
\end{observation}
\begin{proof}
  Suppose we delete from the $\that_j$ all edges from $E_\ell$.
  This will yield $k + |E_\ell|$ trees.
  We will count only those trees that lie in a large component $L_i$, of which there are at most $k + |E_\ell| - 2k_s + k_{s,1} + k_{s,2}$.
  This is because every small component contains at least $2$ of the $\that_j$, except for $k_{s,1} +k_{s,2}$ many.
  Now observe that to get a $(k + |E_\ell| - 2k_s + k_{s,1} + k_{s,2})$-component spanning forest of \emph{minimum} weight for the $L_i$, we start with the initial minimum spanning trees in each component (which have weight at least $(k-k_s-k_\ell)2R$) and remove at most $(k + |E_\ell| -k_\ell - 2k_s + k_{s,1} +k_{s,2})$ edges, each of weight at most $R$. 
  The total remaining weight is 
    $(k-k_s-k_\ell)2R - 	(k + |E_\ell| -k_\ell - 2k_s + k_{s,1} +k_{s,2})
	=(k  - k_\ell - k_{s,1} - k_{s,2} -|E_\ell|)R$. 
  Thus, we can measure the total weight of edges which are part of some $\that_i$ and lie in a large component as being at least the total weight the spanning trees of the~$L_i$, minus the weight of the edges that may have been removed, and obtain
  \begin{align*}
    &\sum_jw(\that_j) - \sum_{\that_i \in T^=} w(\that_i) - |E_\ell|R &&\geq (k  - k_\ell - k_{s,1} - k_{s,2} -|E_\ell|)R .\\
    \implies &\sum_jw(\that_j) - \sum_{\that_i \in T^=} w(\that_i) &&\geq  (k  - k_\ell - k_{s,1} - k_{s,2})R \enspace .\tag*{\qedhere}
  \end{align*}
\end{proof}

Notice that with these two observations, we can almost show some result along the lines of $OPT_{1,k} \geq c \cdot k \cdot R$ for some $c \in \mathbb{R}$, since either there are many large trees, in which case \Cref{obs:l1LargeComponents} gives us the desired result, or there are many small trees in which case we can use \Cref{obs:l1ShortComponents}, unless~$k_{s,2}$ is large.
However, the case that $k_{s,2}$ is large, i.e., we have computed many of the trees that are present in the optimal solution, should also be beneficial, so we maintain a separate record of $k_{s,2}$ in the analysis to obtain:
\begin{lemma}
\label{lem:totalWeight}
  It holds that $\sum_jw(\that_j) \geq \frac{R}{6}(k - k_{s,2}) + \sum_{\that_i \in T^=} w(\that_i)$.
\end{lemma}
\begin{proof}
  We combine \Cref{obs:l1LargeComponents} and \Cref{obs:l1ShortComponents} convexly with coefficients $1/3, 2/3$ to obtain 
  \begin{align*}
    \sum_jw(\that_j) &\geq \frac{1}{3}\left[(k  - k_\ell - k_{s,1} - k_{s,2})R\right]+ \frac{2}{3}\left[ k_{s,1}\frac{R}{2}\right] + \sum_{\that_i \in T^=} w(\that_i)\\
    \iff \sum_jw(\that_j) &\geq \frac{R}{3}(k - k_\ell - k_{s,2}) + \sum_{\that_i \in T^=} w(\that_i)\\
    \implies \sum_jw(\that_j) &\geq \frac{R}{6}(k - k_{s,2}) + \sum_{\that_i \in T^=} w(\that_i),
  \end{align*}
  where the final inequality follows from the following reasoning:
  \begin{equation*}
    k = \sum\left\lfloor \frac{w(S_i) + 2R}{2R} \right\rfloor + \sum\left\lfloor \frac{w(L_i) + 2R}{2R} \right\rceil  \geq k_s + 2k_\ell,
  \end{equation*}
  and thus
   $ k-k_\ell -k_{s,2} \geq k - \frac{k-k_s}{2} -k_{s,2} = \frac{k}{2} + \frac{k_s}{2} - k_{s,2} \geq \frac{k-k_{s,2}}{2} $.
\end{proof}
The upshot of \Cref{lem:totalWeight} is then  that any tree cover using $k$ trees with some $k_{s,2}$ trees in common with the algorithm's solution will have the property that the other $k-k_{s,2}$ trees have an average weight of~$\Omega(R)$.

\begin{theorem}
\label{thm:noDepotwithexactlyK}
  If the algorithm of Even et al.\ returns $k$ trees $T_1,\dots,T_k$, then we have
  \begin{equation*}
    \left[\sum_{i=1}^t (w(T_i)^p)\right]^{1/p} \leq 24\left[\sum_{i=1}^k (w(\that_i)^p)\right]^{1/p}
  \end{equation*}
  for any $p$ and for any tree cover $\{\that_1,\hdots,\that_k\}$ of $G$ with $k$ trees.
\end{theorem}
\begin{proof}
  From \Cref{lem:totalWeight} we obtain
  \begin{equation*}
    \sum_{i=1}^k (w(\that_i)^p) \geq (k-k_{s,2})\left(\frac{R}{6}\right)^p + \sum_{\that_i \in T^=} w(\that_i)^p, 
  \end{equation*}
  using convexity of $|x|^p$ for any $p \geq 1$.
  At the same time, from the algorithm it is clear that
  \begin{equation*}
    \sum_{i=1}^k (w(T_i)^p) \leq (k-k_{s,2})(4R)^p + \sum_{\that_i \in T^=} w(\that_i)^p,
  \end{equation*}
  because every tree is either identical to one of the $\that_j$, or has weight at most $4R$.
  Putting these inequalities together yields the statement of the theorem:
  \begin{align*}
    \frac{\sum_{i=1}^k (w(T_i)^p)}{\sum_{i=1}^k (w(\that_i)^p)} &\leq \frac{ (k-k_{s,2})(4R)^p + \sum_{\that_i \in T^=} w(\that_i)^p}{(k-k_{s,2})\left(\frac{R}{6}\right)^p + \sum_{\that_i \in T^=} w(\that_i)^p}
		\leq \frac{ (k-k_{s,2})(4R)^p }{(k-k_{s,2})\left(\frac{R}{6}\right)^p}
		\leq 24^p \enspace. \tag*{\qedhere}
  \end{align*}
\end{proof}

The whole proof, in particular the result of \Cref{lem:totalWeight} can be reinterpreted to give an even stronger result: namely, that the tree cover returned by the algorithm is approximately ``strongly optimal'', a concept introduced by Alon et al. \cite{alon1998approximation} for analysing all-norm scheduling algorithms.
The point is to show a strong version of lexicographic minimality of the constructed solution. 
Formally, let $\{T_1,\dots,T_k\}$ be the solution computed by the algorithm for a given instance $(\ell_{p\in[1,\infty)},k)$ and $\that_1,\dots,\that_k$ any other tree cover for $(\ell_{p\in[1,\infty)},k)$.
Further, let the trees be sorted non-increasingly by weight as $w(T_1)\geq w(T_2)\geq \dots$ and $w(\that_1)\geq w(\that_2)\geq \dots$.
Then $\{T_1,\hdots,T_k\}$ is \emph{strongly optimal} if for any~$j$ we have
\begin{equation*}
  \sum_{i=1}^jw(T_i) \leq \sum_{i=1}^jw(\that_i) \enspace .
\end{equation*}
Similarly, one can speak of \emph{$c$-approximate strongly optimality} if for some $c\in \reals_{\geq 1}$ we have 
\begin{equation*}
  \sum_{i=1}^jw(T_i) \leq c\sum_{i=1}^jw(\that_i) \enspace.
\end{equation*}
This property was also considered as ``global $c$-balance'' by Goel and Meyerson \cite{goel2006simultaneous}.

Approximate strong optimality suffices for our purposes, since a solution that is $c$-approximate strongly optimal is also a $c$-approximation with respect to any $p$-norm; in fact, we obtain a stronger result here, since $c$-approximate strongly optimality implies a $c$-approximation with respect to \emph{any} convex symmetric function, including all monotone symmetric norms.
This fact is the backbone of many all-norm approximation algorithms, for example those by Golovin, Gupta, Kumar and Tangwongsan~\cite{golovin2008all} for set cover variants.
Thus, we will obtain a $24$-approximation from this analysis not only for all $p$-norms, but for all monotone symmetric norms.

\begin{lemma}
\label{lem:apxStrongOpt}
  Let $\{T_1,\dots,T_k\}$ be the solution computed by the algorithm for a given instance $(\ell_{p\in[1,\infty)},k)$, and let $\{\that_1,\dots,\that_k\}$ be any other tree cover for $(\ell_{p\in[1,\infty)},k)$, where $w(T_1)\geq w(T_2)\geq \dots$ and $w(\that_1)\geq w(\that_2)\geq \dots$.
  Then
    $\sum_{i=1}^jw(T_i) \leq 24\sum_{i=1}^jw(\that_i)$ for $j=1,\hdots,k$. 
\end{lemma}
 \begin{proof}
   We obtain an upper bound for the values $\sum_{i=1}^jw(T_i)$ by assuming that all trees which are not accounted for by the $k_{s,2}$ small component trees shared between the $T_i$ and the $\that_i$ have weight exactly~$4R$.
   This will ensure that the shared trees are $T_{k-k_{s,2}+1},\dots,T_k$.
   Similarly, we obtain a lower bound $\sum_{i=1}^jw(\that_i)$ by allowing a non-descending permutation of the $\that_i$. 
   That is, we reorder the $\that_i$ such that the first $k-k_{s,2}$ trees are not the shared small component trees with the $T_i$.
 	
   It then follows directly that
   \begin{align*}
     \sum_{i=1}^j w(T_i) \leq j\cdot4R = 24 (j\cdot\frac{R}{6}) \leq 24\sum_{i=1}^j w(\that_i) \text{ for } j \leq k-k_{s,2},
   \end{align*}
   as well as 
   \begin{align*}
     \sum_{i=1}^{k-k_{s,2}} w(T_i) + \sum_{i=1+k-k_{s,2}}^{j}w(T_i) &\leq  24\sum_{i=1}^{k-k_{s,2}} w(\that_i) + \sum_{i=1+k-k_{s,2}}^{j}w(\that_i)\\
     \implies
     \sum_{i=1}^j w(T_i) &\leq  24\sum_{i=1}^j w(\that_i) \text{ for } j>k-k_{s,2} \enspace .\tag*{\qedhere}
  \end{align*}
\end{proof}
 
\subsection{Ensuring \texorpdfstring{$k$}{k} trees}
\label{ssec:exactlyKtrees}
Recall that the previous analysis of the approximation factor relies on the algorithm being able to find some $R$ such that \Cref{alg:item:exactlyK} holds, which is not generally true as per \Cref{fig:CounterexampleMinMaxAlgo}. 
We now demonstrate how to modify the algorithm in such a way that this is avoided.
Consider a list $e_1,\dots,e_m$ of the edges of $G$, such that $d(e_i) \leq d(e_j)$ if $i \leq j$, i.e., they are sorted by length with ties broken arbitrarily.
Then we consider separately the graphs $G_j := (V,\{e_i\mid i \leq j \})$ for all $j = 0,\dots,m$ and try to find for each~$G_j$ an $R \in [ d(e_{j-1}),d(e_{j}) ]$ that is accepted by the algorithm, where we set $d(e_0):=0$, and allow the larger interval $[d(e_{m}),  \sum_i d(e_i)]$ for $G_m$.
Note that the intervals can potentially only contain a single node, but they are not empty.

Now observe that for $G_0$ with $R=0$, the algorithm would only accept this choice of $R$ if $k=|V|$. 
Similarly, for $G_m$ and $R = \sum_i d(e_i)$, the algorithm would require $k = 1$.
We can then show that the $k$ changes by at most one as we change $R$ or keep $R$ and move from $G_j$ to $G_{j+1}$.
Let $k(G_j, R)$ denote the value of $k$ that would be accepted by the algorithm.
Thus
\begin{observation}
  It holds that $k(G_{j-1},d(e_j)) - k(G_{j},d(e_j)) \leq 1$.
\end{observation}
\begin{proof}
  Between $G_{j-1}$ and $G_j$, only the presence of $e_j$ changes.
  If this does not change the connected components, we have $k(G_{j-1},d(e_j)) = k(G_{j},d(e_j))$.
  Otherwise, there is exactly one connected component $C$ in $G_j$ that is split into two parts $C_1,C_2$ by removing $e_j$.
  We then see that 
  \begin{align*}
    \left\lfloor \frac{w(C_1)+2d(e_j)}{2d(e_j)} \right\rfloor + \left\lfloor \frac{w(C_2)+2d(e_j)}{2d(e_j)} \right\rfloor &\leq 2+  \left\lfloor \frac{w(C_1)}{2d(e_j)} + \frac{w(C_2)}{2d(e_j)} \right\rfloor \leq 1+ \left\lfloor \frac{w(C) + 2d(e_j)}{2d(e_j)} \right\rfloor \enspace .\tag*{\qedhere}
  \end{align*}
\end{proof}

To see that changing $R$ by a sufficiently small amount also only changes $k(G_j, R)$ by at most one, consider that there are only finitely many critical nodes where $k(G_j, R)$ changes at all, and they can be computed in polynomial time. 
They are all of the form $R = w(C_i) / 2\ell$ for some connected component $C_i$ of $G_j$ and $\ell \in 1,\dots, n$.
At these nodes, $w(C_i)/2R$ will be an integer, so $\lfloor w(C_i)/2R \rfloor = 1+ \lfloor (w(C_i)-\varepsilon)/2R \rfloor$.
If all critical nodes are pairwise different, we can just iterate over them to find some $G_j$ and $R$ with $k(G_j,R) = k$.

If on the other hand a node is critical for multiple different components, assign to each component a distinct weight which can be taken to be arbitrarily small, ensuring that all critical nodes are now different. 
In effect, this is equivalent to taking all components where $\frac{w(C_i) + 2R}{2R}$ is an integer and allowing the expression $\lfloor \frac{w(T_i) + 2R}{2R}\rfloor$ to also take the value $\frac{w(T_i)}{2R}$.
One may check that this does not impact the analysis, since in the analysis we only need that each component has a minimum spanning tree of weight at least $\lfloor \frac{w(T_i)}{2R} \rfloor \cdot 2R$.
However, for legibility reasons we will suppress this and generally assume that some $R$ can be found with $k(G_j, R) = k$.
Combining with \Cref{thm:noDepotwithexactlyK}, this completes the proof of \Cref{thm:TCnodepots}.

\section{All-norm tree cover problem with depots}
The algorithm for the case \emph{with} depots is considerably more involved, in particular because the simple lower bound for the depot-less algorithm of assuming an even distribution of the total weight can no longer work: it might be necessary to have unbalanced cluster sizes, for instance if some depots are extremely far from all nodes.

Instead our algorithm constructs a  $c$-approximately (here, c is a constant) strongly optimal solution by also maintaining some (implicit) evidence that the optimum solution contains large trees if it decides to create a large tree itself. More concretely we do the following:
\begin{enumerate}
  \item First, we partition the node set $V$ into layers $L_i$ such that all nodes in $L_i$ have distance between $2^{i-1}$ and $2^i$ to the depots.
  \item Then we consider separately the nodes in the odd and even layers; this implies that nodes in different layers have a large distance to each other. Computing separate solutions for these two subinstances loses at most a factor $2$ in the approximation factor due to \Cref{lem:steinerMST}.
  \item Next, we partition the nodes in each layer $L_i$ using the non-depot algorithm with $R=2^i$.
  This yields a collection of subtrees in $L_i$ such that the cost of connecting such a subtree to its nearest depot is in $\Theta(2^i)$. 
  This allows us to treat them basically identically, loosing only the constant in the $\Theta$.
  Indeed, we can show explicitly that this prepartitioning into subtrees can be assumed to be present in an optimal solution, up to a constant factor increase in the weight of each tree.
  For the full reasoning, compare the proof of \Cref{lem:alignment} in \Cref{sec:TCDepots}.
  \item To assign these trees to the depots, we iteratively maintain an estimate of the largest tree necessary in any solution as $2^i$.
  We then collect all trees of weight $\Theta(2^i)$ and compute a maximum matching between them and the depots at distance $\Theta(2^i)$ to them.
  All unmatched trees are then combined to form trees of weight $\Theta(2^{i+1})$, and we update our estimate to~$2^{i+1}$.
\end{enumerate}

To analyse the output of this algorithm, it will suffice to show that the estimate was correct up to a constant factor. 
That is, if in round $i$ some trees (suppose $k_i$ trees) of weight $2^i$ were assigned, we will be able to prove that any tree-cover solution must also have some family of at most $k_i$ trees with total weight at least $k_i\cdot 2^i$ (see \Cref{lem:matching} for a formal statement of this criterion). 

From this we are then able to conclude $c$-approximate strong optimality for some value of $c <10^6$. 
This is a rather large upper bound on the approximation factor, however, we conjecture that the actual constant achieved by the algorithm is much smaller. 
For the purposes of a clean presentation, we did not try to optimize the constant in our proof.

\section{Computational Hardness}
To complement our algorithmic results, we establish hardness results for all-norm tree cover problems and discuss on what kind of algorithmic improvements (to our algorithms) are potentially possible in light of these complexity results.
Specifically, we show:

\begin{enumerate}
  \item For any $p\in (1,\infty]$, problem $\ell_p$-{\sc Tree Cover with Depots} is weakly $\mathsf{NP}$-hard, even with only $2$ trees. \label{item:HardnessFromPartition}
  \item For any $p\in (1,\infty]$, problem $\ell_p$-{\sc Tree Cover with Depots} with $k$ trees is (strongly) $\mathsf{W}[1]$-hard parameterized by depot size $|D|$.
  \item For any $p\in (1,\infty]$ there exists some $\varepsilon>0$ such that $\ell_p$-{\sc Tree Cover with Depots} is $\mathsf{NP}$-hard to approximate within a factor $<1+\varepsilon$ under randomized reductions (see \Cref{thm:apxHardness} for the reduction). \label{item:apxHardness}
\end{enumerate}
Results 1 and 2 were already essentially presented by Even et al. \cite{even2004min}, but we restate them for completeness.
Note that, given these complexity results, numerous otherwise desirable algorithmic outcomes become unattainable. 
For instance, there cannot be polynomial-time approximation schemes for $\ell_p$-{\sc Tree Cover with Depots}, nor can we expect fixed-parameter algorithms (parameterized by the number of depots) finding optimal solutions, unless established hardness hypotheses (\PP$\not =$\NP, ${\sf FPT}\not ={\sf W[1]}$) fail. 
Further, the reduction of \Cref{item:HardnessFromPartition} yields bipartite graphs of tree-depth $3$, so parameterization by the structure of the graph supporting the input metric also appears out of reach.

The result in \Cref{item:apxHardness} follows by a direct gadget reduction from {\sc Max-Sat} for $3$-ary linear equations modulo $2$ (MAX-E3LIN2), which was shown to be \APX-hard by H\aa stad \cite{Hastad99Inapx}.
We show that one can transform such equations into an instance of $\ell_p$-{\sc Tree Cover with Depots} where every unsatisfied equation will correspond roughly to a tree of above average weight.
This allows us to recover approximately the maximum number of simultaneously satisfiable constraints of such systems of equations.

Formally, we reduce from $3R\{2,3\}L2$, a modification of  MAX-E3LIN2 where every variable occurs in exactly 3 equations, and for which hardness of approximation was shown by Karpinski et al. \cite{Karpinski2015ApxTSP}.
From an instance of {\sc 3R\{2,3\}L2} we construct an instance of $\ell_p$-{\sc Tree Cover with Depots} by introducing gadgets for the variables and clauses:
\begin{itemize}
  \item For every variable $x$, introduce three nodes $\hat{x}, x_0,$ and $x_1$, as well as edges ${\hat{x}, x_i}$ for $i=0,1$ with weight $3$.
  Add the $x_i$'s as depots.
  \item For every ternary clause $C$, introduce nodes $\hat{C}, C_{000}, C_{110}, C_{101},$ and $C_{011}$ with edges $\{\hat{C}, C_i\}$ of weight $3$.
    Add the $C_i$'s as depots.
  \item For every binary clause $C = x\oplus y = 0$, introduce nodes $\hat{C}, C_{00},$ and $C_{11}$ with edges $\{\hat{C}, C_i\}$ of weight $2$.
    Add the $C_i$'s as depots. Further add nodes $\hat{C_{00}}$ and $\hat{C_{11}}$ where $\hat{C_i}$ is connected only to $C_i$ by an edge of weight $1$.
  \item For every binary clause $C = x\oplus y = 1$, introduce nodes $\hat{C}, C_{01},$ and $C_{10}$ with edges $\{\hat{C}, C_i\}$ of weight $2$.
    Add the $C_i$'s as depots.
    Further add nodes $\hat{C_{01}}$ and $\hat{C_{10}}$ where $\hat{C_i}$ is connected only to $C_i$ by an edge of weight $1$.
\end{itemize}

We connect the gadgets together as follows:
\begin{itemize}
  \item For every clause $C = x\oplus y \oplus z = 0$ we connect $C_{b_1b_2b_3}$ to $x_{b_1}, y_{b_2},$ and $z_{b_3}$ with a path of length $2$ where both edges have weight $1$.
  \item For every clause $C = x\oplus y = b$ we connect $C_{b_1b_2}$ to $x_{b_1}, y_{b_2}$  with a path of length $2$ where both edges have weight $1$.
\end{itemize}

An illustration of this construction is given in \Cref{fig:ApxGadgets}. 
Intuitively there is a depot for every way a clause could be satisfied, and the depots for a clause have a joint neighbour that is expensive to connect. 
The depot absorbing this neighbour should correspond to the way in which the clause is satisfied.
Similarly there are two depots for each variable representing the two possible assignments to the variable.
In this case the depot that does not get assigned the joint neighbour corresponds to the chosen assignment.

There are then additional vertices between the clause and vertex depots which can be assigned to the clause depots, unless the adjacent clause depot corresponds to the satisfying assignment of that clause.
In that case they need to be assigned to an adjacent variable depot, which is to say the variables of the clause will have to be assigned in such a way that the clause is satisfied.

One can quickly check that a satisfying assignment to the clauses will allow us to compute a tree cover where every tree has size $3$.
Meanwhile every non-satisfied clause will push at least one unit of excess weight to a tree of size at least $3$. 
Quantifying this relationship then permits us to compute approximately the maximum number of satisfiable clauses in such a system of equations, see \Cref{sec:Hardness} for details.

\section{Computational Experiments}
\label{sec:experiments}
To illustrate the practical performance achievable by our clustering algorithms, we implemented
the algorithm from \Cref{sec:TCnodepots} for the setting without depots in {\sc C++} and tested it on instances proposed by Zheng et al. \cite{zheng2007Robot}.
Those instances model real-world terrain, which is to be partitioned evenly so that a fixed number robots can jointly traverse it.
They consist of a grid where either random cells are set to be obstacles, i.e., inaccessible, or a grid-like arrangement of rooms is superimposed with doors closed at random. 
A metric is induced on the accessible cells of the grid by setting the distance of neighbouring cells to be $1$.

For all instances on grids of size $200$ by $200$ (ca. 40,000 nodes), our implementation was able to compute a clustering in less than $200$ms on an Intel i5-10600K under Windows with $48$GB of available memory (although actual memory usage was negligible).
The resulting partitions are illustrated in \Cref{fig:Experiments}.
Note that we make two small heuristic changes to the original algorithm, which are, however, not amenable to be analyzed formally:

First, we adapt the partitioning of the large trees in \Cref{alg:item:Decomposition} to try to cut the trees into subtrees of weight at least $2R$ rather than $4R$.
Note that the choice of $4R$ captures a worst case scenario where the instance contains edges of size almost $R$ that are to be included in a solution.
If each of the heaviest edges is considerably lighter than $R$, the necessary cutoff approaches $2R$ rather than $4R$.
    Thus, we try to run the partitioning algorithm with $2R$ rather than $4R$, and resort to the higher cutoff only in the case that this fails.
    However, for the considered instances this was not necessary.
    
Second, we post-process the computed solution to ensure that it has exactly $k$ components.
  It is in principle possible that the algorithm computes a solution with fewer than $k$ trees; in this case, we iteratively select the largest tree and split it into two parts of as similar a size as possible until we obtain exactly $k$ components.
  
The clusterings obtained in this way in \Cref{fig:Experiments} are all at least $3$-approximately strongly optimal when compared to a hypothetical solution that distributes the total weight perfectly evenly on the $k$ clusters.

\begin{figure}[ht]
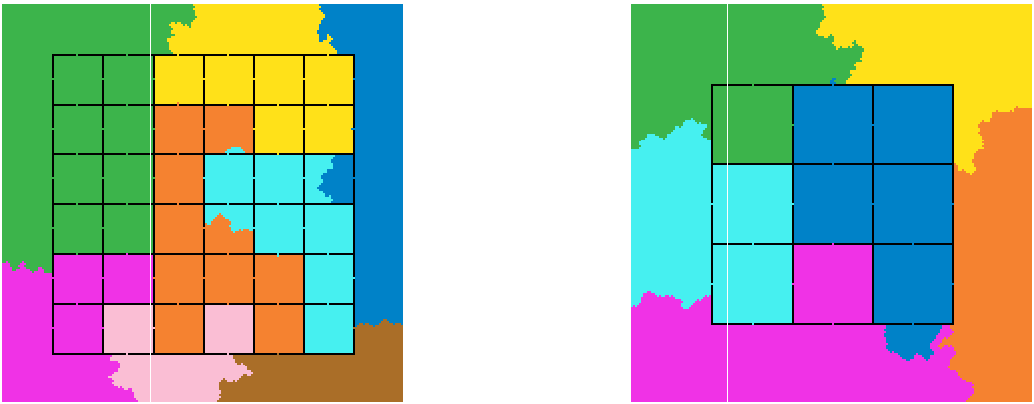

  \centering
  \begin{subfigure}[b]{0.48\textwidth}
    \centering
    \includegraphics[scale=0.022]{WallGrid200_p20_k8.pdf}
    \caption{Output of the algorithm on a $200\times 200$ grid with walls partitioned into $8$ clusters. The cluster sizes are $2433,5516,9528,4550,5271,3985,2482$ and $4240$; consequently the solution is at least $2.01$-approximately strongly optimal.}
  \end{subfigure}
  \hfill
  \begin{subfigure}[b]{0.48\textwidth}
    \centering
    \includegraphics[scale=0.022]{WallGrid200_p40.pdf}
    \caption{Output of the algorithm on a $200\times 200$ grid with walls partitioned into $6$ clusters. The cluster sizes are $7993, 4774, 6390, 8004, 5267,$ and $6638 $; consequently, the solution is at least $1.61$-approximately strongly optimal.}
  \end{subfigure}
  \caption{
  Visualisations of the results of our implementation of the non-depot tree cover algorithm. 
  Inaccessible sections of the grid are marked in black; other colors represent the computed clusters.\label{fig:Experiments}}
\end{figure}

\clearpage
\pagebreak
\nocite{DBLP:conf/iros/ZhengJKK05}
\bibliography{Bibliography}

\clearpage
\pagebreak

\appendix

\section{All-Norm Tree Cover Problem with Depots}
\label{sec:TCDepots}
We now introduce depots to our problem and try to maintain an all-norm approximation. 
In fact, we will also here show $c$-approximate strong optimality of a solution, so this result too can be extended to all monotone symmetric norms.
The basic idea is quite similar to the non-depot case; We will start by partitioning the metric space into trees of roughly equal weight $R$, or, if that is not possible, some smaller trees that are however at distance $R$ to all other trees.

However, since now we need to connect to the depots this is no longer so straightforward. 
Nodes which are far from all depots will inherently cause a high cost, no matter where we assign them.
To then keep the balance between tree weight and distance between trees, we need to categorize the nodes according to their distance to the depots.
Specifically, we partition~$V$ into \emph{layers} $L_0,L_1,\dots$, where layer $L_i$ contains all $v\in V$ for which $d(v,D) \in [2^i,2^{i+1})$. 
We then partition each layer individually with the algorithm of Even et al. \cite{even2004min}, but using increasing values of $R_i = 2^i$ in layer $L_i$.

To simplify things we adapt the algorithm slightly. Layer $i$ will be partitioned according to the following scheme:
\begin{enumerate}
  \item Set $R_i = 2^i$.
  \item Remove all edges with weight greater than $R_i$ and compute a minimum spanning tree $T_C$ for each connected component $C$ of the resulting graph.
  \item Call a spanning tree $T_C$ \emph{small} if $w(T_C) < 2R_i$.
  \item Call a spanning tree $T_C$ \emph{large} if $w(T_C) \geq 2R_i$.
    Decompose any such tree into edge-disjoint subtrees $\Tilde{T}_j$ such that $2R_i \leq w(\Tilde{T}_j) \leq 6R_i $. \label[step]{alg:item:DepotDecomposition}
  \item Output the collection of all small spanning trees, as well as all subtrees $\Tilde{T}_j$ created in \Cref{alg:item:DepotDecomposition}.
\end{enumerate}
Here the decomposition in \Cref{alg:item:DepotDecomposition} avoids the creation of small trees from large component by adding such a small tree to a neighbouring large tree if it is necessary.
This worsens the constant we can get, but it simplifies analysis considerably.

By running this decomposition for each layer separately, we obtain a tree cover $\tset_i$ for each layer $i$.
However, notice that one critical property of the no-depot setting is missing here: small component trees may not be far from all other nodes, they are only far from nodes in their layer.
We sidestep this challenge by separating the instance into two parts, the layers of even and odd indices.
For $v\in L_i$ and $w \in L_{i+2}$ we are guaranteed that $d(v,w) \geq d(w,D) - d(v,D) \geq 2^{i+2} - 2\cdot 2^i = 2^{i+1}$ from the triangle inequality.

We thus treat the instance induced by the odd-layer nodes, and the instance induced by the even-layer nodes, as two separate instances. 
This incurs a factor of $2$ in the approximation factor since we now combine two unrelated instances, and an additional factor $2$ because we lose some points of the metric space so we are constricted to finding spanning trees on the restricted space rather than Steiner trees in the whole space, see \Cref{lem:steinerMST}.
So for all of the following we assume that all nodes are in even layers, where the analysis for all nodes being in even layers is fully analogue (see \Cref{fig:AssignmentAlgorithmPartition} for illustration).

\begin{figure}[ht]
  \centering
  \begin{subfigure}[b]{\textwidth}
    \centering
    \includegraphics[scale=0.5,page=1]{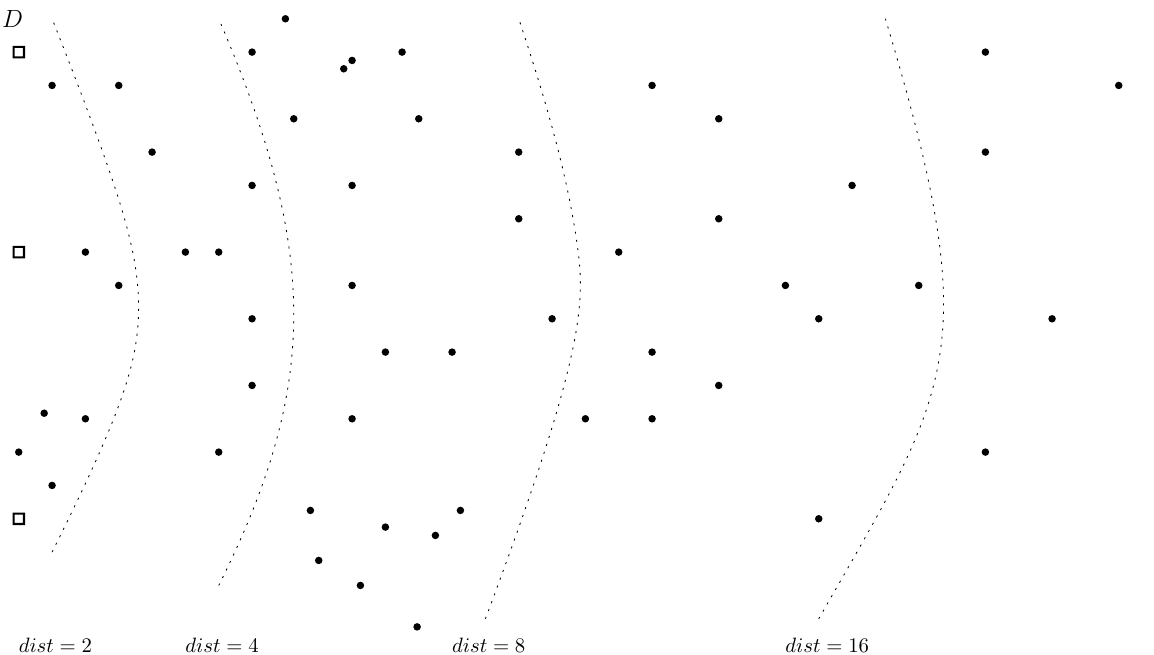}
    \caption{Initial instance, with the depots indicated as boxes and the nodes as dots. Nodes are precategorised by their distance to the depots.\label{fig:AssignmentaAlg1}}
  \end{subfigure}
  \\
  \begin{subfigure}[b]{\textwidth}
    \centering
    \includegraphics[scale=0.5,page=2]{Assignmentillustration.pdf}
    \caption{In the first step we discard every second layer and cluster each remaining layer independently. Connected components after removing edges of length $>R_i$ are indicated by dashed lines, the trees extracted within them by solid lines .\label{fig:AssignmentAlg2}}
  \end{subfigure}
  \caption{Illustration of the initial partitioning step of our algorithm.\label{fig:AssignmentAlgorithmPartition}}
\end{figure}

For this partitioning of the space into the trees from $\{\tset_i\}$ we can show that any depot-based tree cover $\{\that_o\}_{o\in D}$ can be restructured as $\{\that_o'\}_{o\in D}$ so as to agree with the partitions, i.e., such that for any $i$ and any $T\in \tset_i$ there is a $o\in D$ for which $V(T) \subseteq V(\that_o')$ and $w(\that_o) \leq c\cdot w(\that_o')$.

\begin{lemma}
\label{lem:alignment}
  Let $\{\that_o\}_{o\in D}$ be any tree cover with depots and $\{\tset_i\}_i$ the trees computed in layer $i$ by the partitioning algorithm. 
  There exists a tree cover with depots $\{\that_o'\}_{o\in D}$ such that for any $i$ and any $T\in \tset_i$ there is a $o\in D$ for which $V(T) \subseteq V(\that_o')$ and $w(\that_o) \leq 48\cdot w(\that_o')$.
\end{lemma}

\begin{proof}
  We will assign every $T_i$ to some $\that_o$ in a way that does not increase the weight of any tree too much.
  Consider first a $T_i$ generated from some small component $C$ in layer $L_j$, and take~$\that_o$ to be some tree containing a node from $T_i$. 
  Then $w(T_i) \leq 2\cdot 2^j$, and $\that_d$ contains an edge entering~$C$.
  All those edges have weight at least $2^j$ if they are contained in layer $j$ or coming from layer $j+2$, and weight at least $2^{j-1}$ if they are coming from layer $j-2$. 
  Thus we can assign~$T_i$ to $\that_d$ and charge the extra cost against that edge, loosing a factor $8$ at most.
	
  For the large components, we need a more careful argument. 
  Let $C$ be such a component, let $\{F_i\}_i$ the collection of subtrees cut from~$C$, and let $\{\ttil_o\}_o$ be the set of trees $\that_o$ containing a node from $C$.
  For each $\ttil_o$ we look at its edges $X_o := \{e\in E(\ttil_o \mid e\cap C \not = \emptyset\}$ which are incident to $C$. 
  Then we define a capacity of such a tree as $c(\ttil_o) = \lfloor\lambda d(X_o) / 2^j \rfloor$, where we will choose $\lambda$ in a moment.
  This capacity represents how many of the $F_i$ can be absorbed by $\ttil_o$.
  Then for any constant choice of $\lambda$ assigning at most $c(\ttil_o)$ of the $F_i$ to $\ttil_o$ will only increase the weight of $\ttil_o$ by a constant factor $6\cdot \lambda$, charged against $X_o$.
	
  To show that such an assignment exists, we employ Hall's Theorem. 
  We allow any $F_i$ to be assigned to a $\ttil_o$ with which it shares a node, so we construct the bipartite graph of this incidence system and verify Hall's criterion.
  Let $\{F_i\}_{i\in A}$ be some subset of the $F_i$, and $\{\ttil_o\}_{o\in B}$ the set of $\ttil_o$ incident to them.
  We must then show that $\sum_{o\in B} c(\ttil_o) \geq |A|$.
  To do this define $V_A$ as $\bigcup_{i\in A}V(F_i)$, and consider the number $k$ of connected components we obtain after removing from $\{\ttil_o\}_{o\in B}$ all edges entering $C$, and counting only those components intersecting $V_A$.
  Then for each such connected component, $\{\ttil_o\}_{o\in B}$ contains a Steiner Tree within $C$.
	
  Again here we exploit \Cref{lem:steinerMST}, as well as the fact that for each of the $k$ components there is an edge of cost at least $2^{j-1}$ entering $C$. 
  This allows us to conclude that
  \begin{equation*}
    \sum_{o\in B} d( X_o) \geq k2^{j-1} + \frac{1}{2}\max\{0,\sum_{i\in A}w(F_i)-k2^j\},
  \end{equation*}
  where we also note that if $k \geq |A|$ we could get an assignment with $\lambda = 2$. 
  So we have $k < |A|$, and $w(F_i) \geq 2\cdot 2^j$, and thus obtain that
  \begin{equation*}
    \sum_{o\in B} d( X_o) \geq k2^{j-1} + |A|2^{j-1} = (|A|+k)2^{j-1} \enspace .
  \end{equation*}
  This gives Hall's criterion for $\lambda = 4$, concluding our argument.
	
  Therefore, we can match any tree $T$ in the $\{\tset_i\}_i$ to some $\that_o$ while only increasing the cost by a factor at most $24$, charged against some edges of the $\that_o$ incident to $T$. 
  Since every edge gets charged at most twice by this matching, the total increase in cost for each tree is bounded by a factor of $48$.
\end{proof}

Notice that \Cref{lem:alignment} allows us to assume that tree covers are aligned with the partition we compute while only increasing the weight of each tree by a constant factor. 
This in particular means that we obtain a constant-factor approximation for {\sc Tree Cover with Depots} for any norm for which we can obtain a constant-factor approximation after assuming such alignment.

Based on the guarantee of \Cref{lem:alignment} we are now ready to develop an algorithm for assigning the trees in the $\tset_i$ to depots.
To do so we will need some additional notation.
We partition every~$\tset_i$ into sets $\{\tset_{i,o}\}_{o \in D}$ by assigning every $T \in \tset_i$ to a $\tset_{i,o}$ such that $d(T,o)$ is minimal over the choice of $o$.
There is another simplification we can now make. 
For every tree $T$ in some $\tset_{i,o}$ we can add an edge $\{v,o\}$ to $T$ such that $v\in V(T)$ and $d(v,o) = d(T,o)$.
By similar arguments as before this increases the cost of our solutions by at most a factor of $4$ since that additional edge has cost bounded by either~$w(T)$, or the cost of any edge in $\delta(V(T))$.
This causes all trees in~$T_{i,o}$ to have pairwise distance $0$, allowing us to recombine them without additional costs.
Further, it also ensures that \emph{all} trees in $\tset_{i,o}$ have the same cost, up to constant. 
Specifically, any such tree has cost at least $2^i$ from the edge to $o$ and at most $8\cdot 2^i$ if it was a large tree of weight $6\cdot 2^i$ and the edge to $o$ had cost $2^{i+1}$.
Giving up on those constants in the approximation factor, we shall therefore also assume that they all have weight exactly $2^i$.

To recap, after reductions we have a partitioning of the nodes into trees, each of which belongs to some collection $\tset_{i,o}$ where any tree in $\tset_{i,o}$ has weight $2^i$ and is at distance $0$ to~$o$, but possibly at a much larger distance to the other depots.
We further assume that any solution treats these trees as immutable, so in a fixed solution $\{\that_o\}_{o\in D}$ there is for every $T\in \bigcup_{i,o}\tset_{i,o}$ a~$\that_o$ containing~$T$.

We will now try to iteratively match the trees in $\tset_i$ to depots.
To do, so we set a current weight $R_i$, initially $1$.
We then collect all trees of weight $\Theta(R_i)$, and compute a matching between those trees and the depots at distance $\Theta(R_i)$ to them.

Any leftover trees after this step are paired up and treated as a new tree of weight $2R_i$.
We then update $R_{i+1} = 2R_i$ and repeat the process. 
There are some additional details to take care of, but in substance this will be enough to obtain a relatively even distribution of the trees to the depots.
For the full description, we refer to \Cref{alg:assignment}.
For an illustration, see \Cref{fig:AssignmentAlgorithmAssignment}.

\begin{algorithm}[!ht]
  \KwIn{A metric $(V, d)$, depots $D$, trees $\{\mathcal{T}_{i, o}\}_{i,o}$ from the partitioning algorithm. }
  \KwOut{ An assignment of trees to depots}.
	 $i \gets 0,R_i \gets 1$ \;
	 $\mathcal{F}_{i,o} \gets \mathcal{T}_{i,o}$\;
	\While{$\{\mathcal{F}_o\}_{o} \neq \emptyset$ for any $o\in D$}{
		Construct a bipartite matching graph $H_i$ between the trees in $\mathcal F_{i,o}$ and the depots\;
		$V(H_i) := \bigcup_{o\in D}\mathcal{F}_o \dot \cup D$\;
		$E(H_i) := \{\{F, o\} \mid d(F,o) \leq R_i\}$\;
		Compute a maximal matching $M_i$ in $H_i$ and assign every matched tree to the depot it was matched to\;
		Remove all matched trees from the $\mathcal F_{i,o}$\;
		For every $\mathcal F_{i,o}$ with odd cardinality, assign one tree from $\mathcal{F}_{i,o}$ to $o$\;
		For every non-empty $\mathcal{F}_{i,o}$ pair up the trees in $\mathcal{F}_{i,o}$; for every pair $F_1,F_2$ set $\mathcal{F}_{i,o} := \mathcal{F}_{i,o}\setminus \{F_1,F_2\}\cup \{F_1\cup F_2\}$\;
		$\mathcal{F}_{i+1,o} := \mathcal{F}_{i,o} \cup \tset_{i+1,o}$\;
		$R_{i+1} \gets 2R_i, i\gets i+1$\;
	}
  \caption{Assignment algorithm}
\label{alg:assignment}
\end{algorithm}

Let us make some remarks on the algorithm to simplify its analysis:
\begin{itemize}
  \item In line 9, the possible additional assignment of a tree from $F_o$ happens only when $o$ was already assigned another tree in line 7. 
  Thus, the additional cost due to this step is an increase of the tree weights by at most a factor $2$.
  Loosing this factor we may thus assume that line 7 is never applied.
  \item The recombination of trees in line 10 is free due to the inclusion of edges to $o$ for every tree from $\tset_{i,o}$. 
  \item The load of any depot at the termination of the algorithm depends only on the last iteration~$i$ in which it was assigned some tree, i.e., the load will be at most $2\cdot 2^i$ by a geometric sum argument.
\end{itemize}

\begin{figure}[!ht]
  \centering
  \begin{subfigure}[b]{0.45\textwidth}
    \centering
    \includegraphics[width=\textwidth,page=3]{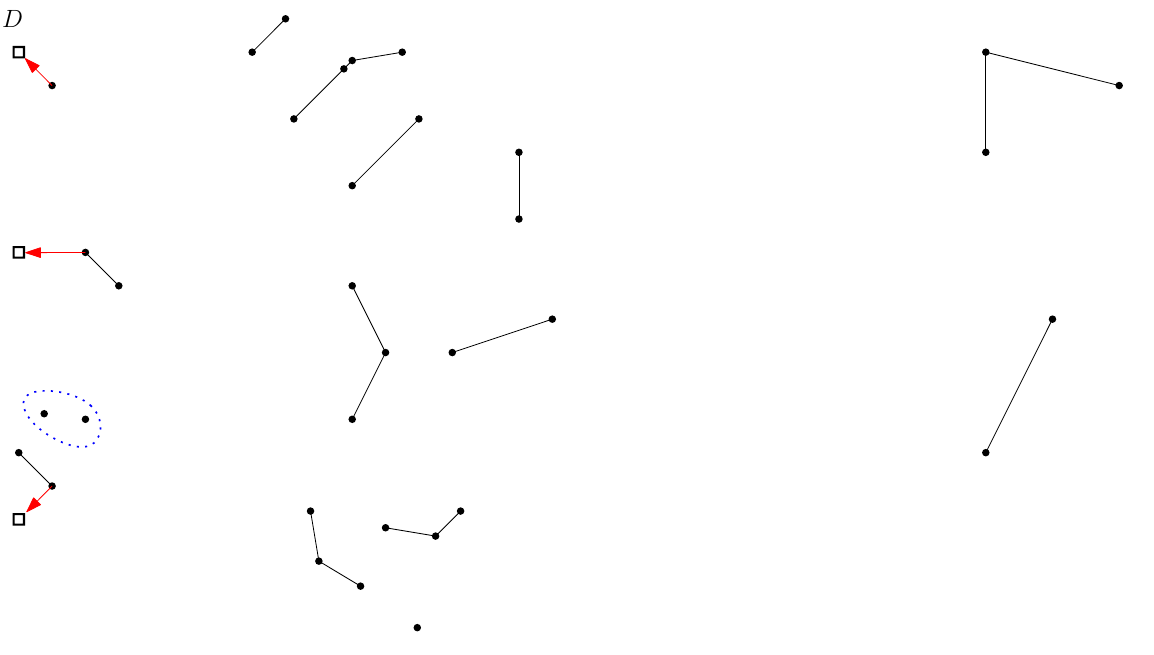}
    \caption{First Round of Assignment.\label{fig:AssignmentaAlg3}}
  \end{subfigure}
  \hfill
  \begin{subfigure}[b]{0.45\textwidth}
    \centering
    \includegraphics[width=\textwidth,page=4]{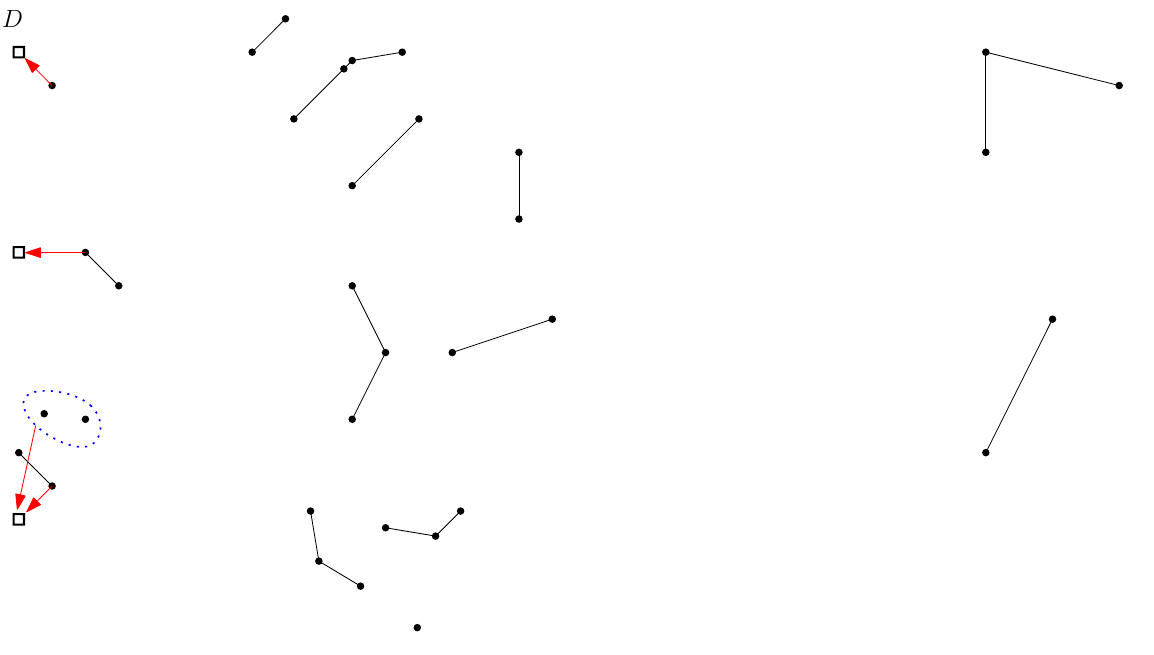}
    \caption{Second Round of Assignment.\label{fig:AssignmentAlg4}}
  \end{subfigure}\\
  \begin{subfigure}[b]{0.45\textwidth}
    \centering
    \includegraphics[width=\textwidth,page=5]{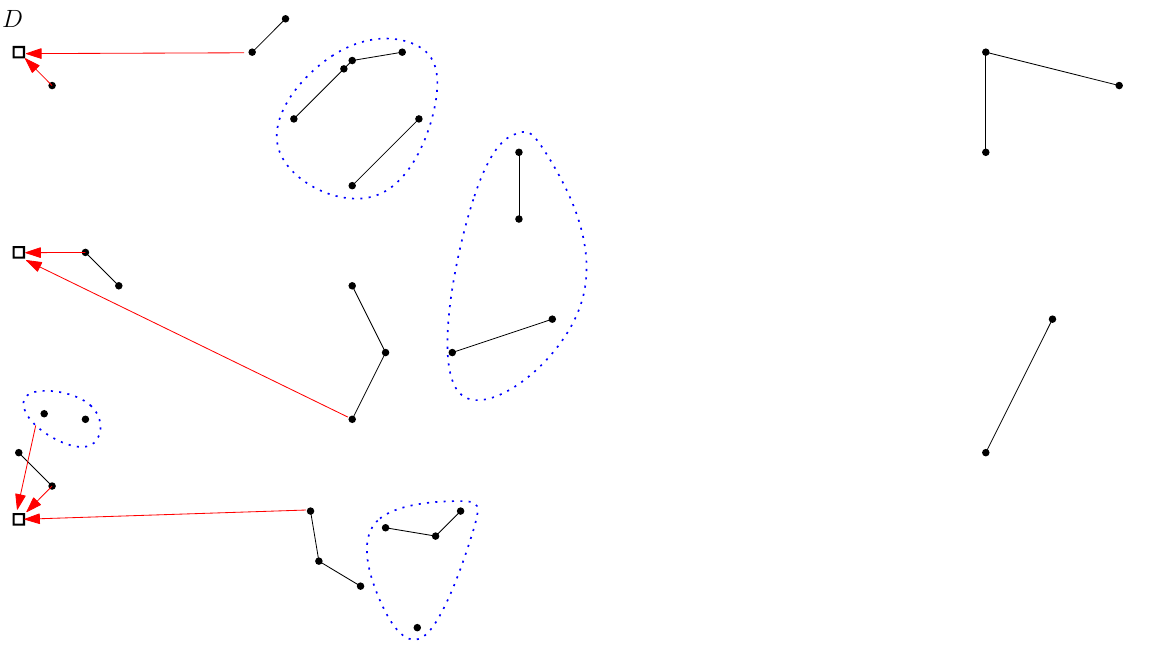}
    \caption{Third Round of Assignment.\label{fig:AssignmentaAlg5}}
  \end{subfigure}
  \hfill
  \begin{subfigure}[b]{0.45\textwidth}
    \centering
    \includegraphics[width=\textwidth,page=6]{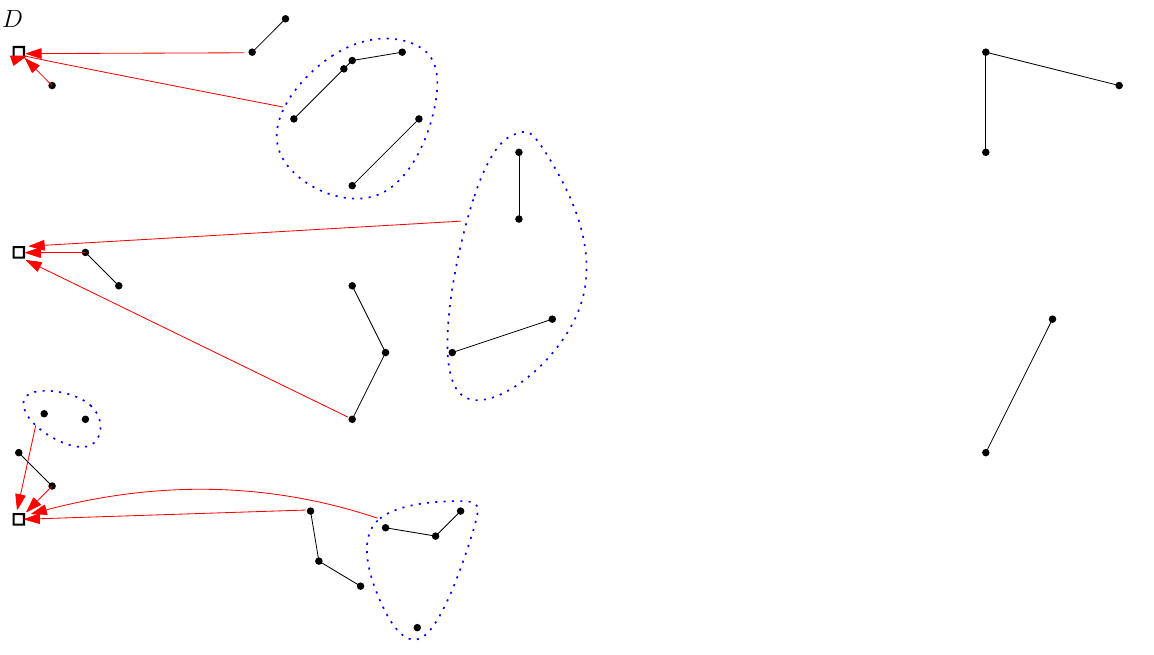}
    \caption{Fourth Round of Assignment.\label{fig:AssignmentAlg6}}
  \end{subfigure}\\
  \subcaptionbox{Fifth Round of Assignment.\label{fig:AssignmentaAlg7}}{\includegraphics[width=0.45\textwidth,page=7]{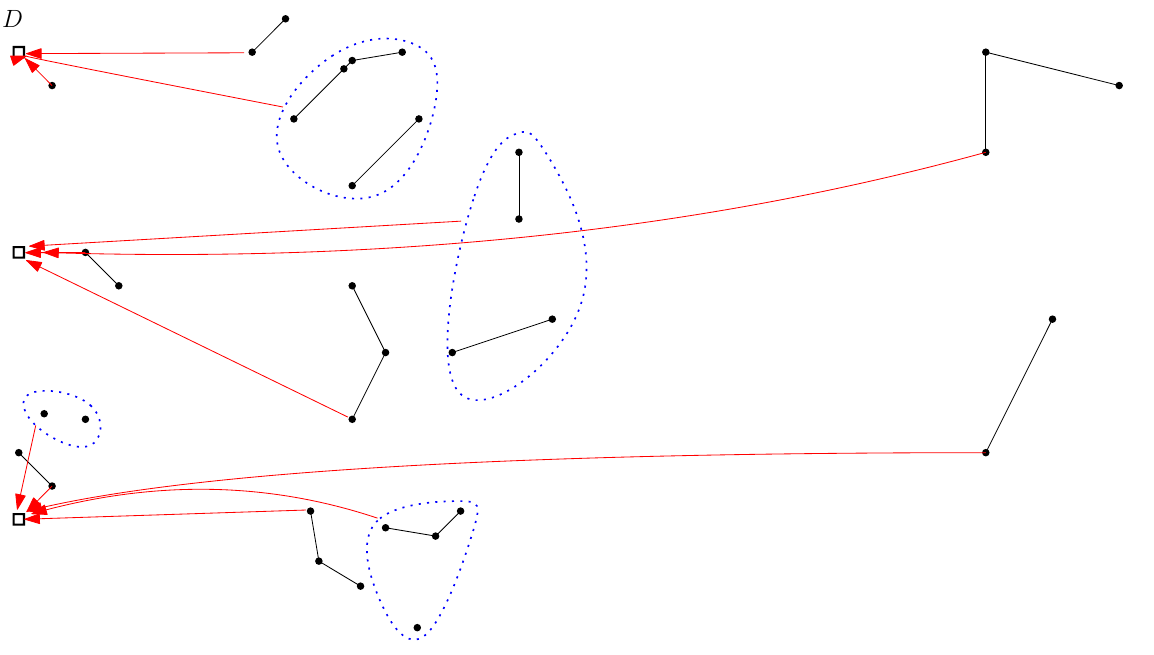}}
  \hfill
  \subcaptionbox{Final output, with minimum spanning trees computed over the assigned trees for every depot.\label{fig:AssignmentAlg8}}{\includegraphics[width=0.45\textwidth,page=8]{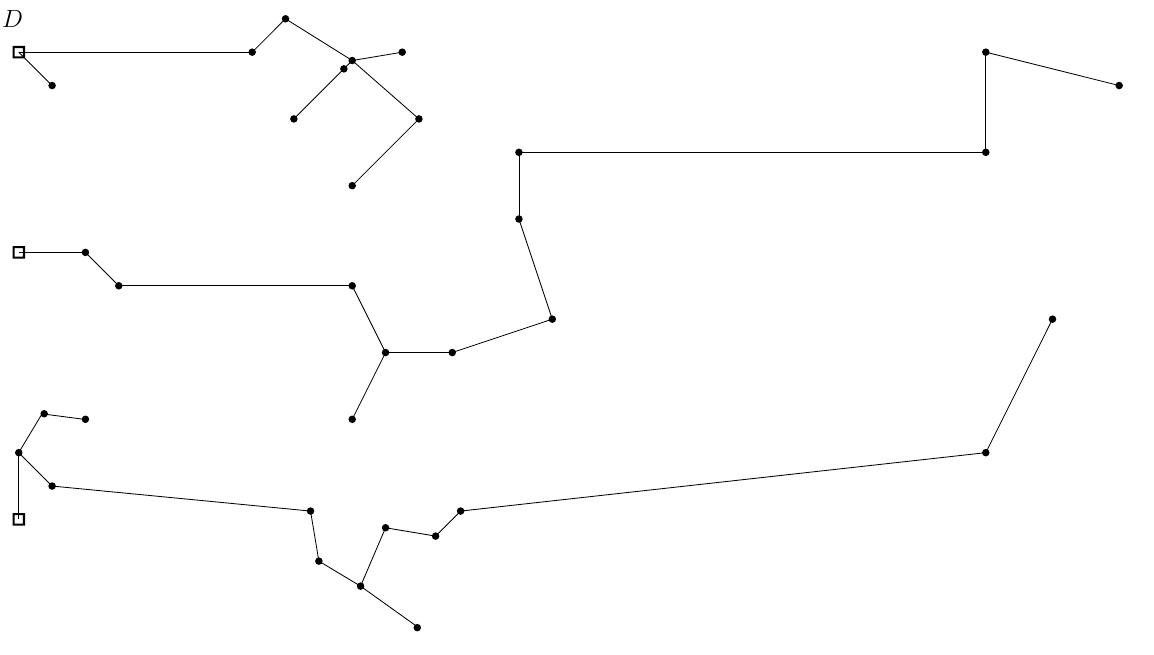}}
  \caption{Illustration of the iterative assignment of trees to depots in the algorithm. Red arrows indicate when trees are assigned to depots; Blue dotted potatoes indicate merged trees.\label{fig:AssignmentAlgorithmAssignment}}
\end{figure}

To compare the output of the algorithm to any other depot-based tree cover $\{\that_o\}_{o\in D}$ we will employ the following scheme:
For each weight $R_i$ we will cut up the trees $\that_o$ into parts of weight $R_i$. 
If enough such parts exist we will then be able to assign every tree of weight~$R_i$ from the algorithms' solution to one such part, thus successfully demonstrating that the algorithms solution provides a ``good'' tree cover.
So fix an arbitrary depot-based tree cover $\hat{\tset}:=\{\that_o\}_{o\in D}$ and denote by $\Lambda_i$ the number of such parts we could get when extracting fragments of weight $R_i$, i.e., $\Lambda_i := \sum_{o\in D} \lfloor \that_o/R_i \rfloor$.
The critical lemma we need then is this:
\begin{lemma}
\label{lem:matching}
  It holds that
  \begin{equation*}
    \Lambda_i \geq \frac{1}{2} \sum_{j={i+1}}^\infty|M_j|2^{j-i},
  \end{equation*}
  where $M_j$ are the matchings computed by the algorithm.
\end{lemma}

\begin{proof}
  We split the analysis into two parts, depending on whether the weight $ \sum_{j={i+1}}^\infty|M_j|2^{j-i}$ is largely caused by trees of weight at least $R_i$ being assigned, or by many trees of weight strictly smaller than $R_i$. 
  This will ultimately cause us to lose the additional factor $\frac{1}{2}$.
	
  First, consider some tree $T$ of weight $R_j$, for $j \geq i$.
  We assume that such a tree is entirely assigned to some $\that_o$ where it contributes $2^{j-i}$ to $\Lambda_i$. 
  This also remains the case if the tree has been recombined and matched at some iteration later than $j$.
  Thus, if we only consider such large trees we do in fact obtain $\Lambda_i \geq \sum_{j={i+1}}^\infty|M_j|2^{j-i}$.
  So we can assume that such trees do not exists and focus instead on the trees which are smaller than $R_i$ but are not yet assigned in iteration $i$.
	
  For these small trees the lemma promises some concentration property. 
  In principle, they are small enough to not contribute to $\Lambda_i$.
  However, we will now show that there are too many such trees, or rather too few depots to accept them.
  Thus, some depot must either contain many of them, or they must be assigned to depots further away. 
  Either way they make the necessary contribution to $\Lambda_i$.
	
  Let $D_i$ be the set of depots $o$ for which $\fset_{i,o}$ was not empty at the end of iteration $i$.
  In particular these depots have been assigned some tree during the iteration.
  Now in $H_i$ let $\hat{D}$ be the set of depots reachable from $D_i$ using an alternating path with respect to $M_i$ (i.e., alternatingly using an edge from~$M_i$ and not from $M_i$).
  By maximality of $M_i$, all depots in~$\hat{D}$ were assigned some tree during this iteration; further, we know that for any tree $T$ in~$\fset_{i,o}$ that is at distance at most $2^i$ to a depot of~$\hat{D}$ we have $d(T,D\setminus \hat{D}) > R_i$.
  Otherwise, we could find an alternating path to some empty depot and augment the matching as before.
	
  So consider $\tset$ to be the set of all trees which are unassigned after iteration $i$ and have weight at least $R_i$, as well as the set of trees assigned to some depot in $\hat{D}$ during iteration $i$.
  Now take $\{\that_o\}_{o\in X}$ to be the set of trees from $\hat{\tset}$ containing any tree from $\tset$.
  We observe that
  \begin{equation*}
    \sum_{o\in A}w(\that_o) \geq \sum_{T\in \tset}w(T) + R_i \cdot |A\setminus \hat{D}|,
  \end{equation*}
  where the additional term  $R_i \cdot |A\setminus \hat{D}|$ originates from the fact that any depot outside of~$\hat{D}$ has distance at least $R_i$ to the trees in $\tset$.
  Now consider that $\sum_{T\in \tset}w(T) = |\hat{D}|R_i + \sum_{j={i+1}}^\infty|M_j|2^{j-i}$ (again, we assumed that no trees of weight larger than $R_i$ are present in the original partition).
  We thus obtain that
  \begin{equation*}
    \sum_{o\in A}\lfloor w(\that_o)/2^i\rfloor \geq \paren[\big]{|\hat{D}|R_i + \sum_{j={i+1}}^\infty|M_j|2^{j-i} +  R_i \cdot |A\setminus \hat{D}| }/ R_i - |A| = \sum_{j={i+1}}^\infty|M_j|2^{j-i}
  \end{equation*}
	
  Notice that this argument induces some double counting on the trees of weight strictly larger than $R_i$, since they are both counted for their weight but also potentially in the distances between small trees and far-away depots.
  Thus, we need to give up a factor $2$ in the $\Lambda_i$ to obtain the final statement of the lemma.
\end{proof}

This lemma now guarantees the existence of sufficiently many trees of some weight $R_i$ in~$\hat{\tset}$ to account for the trees of weight $R_i$ in the algorithms solution. 
The only issue here is the factor of $\frac{1}{2}$ in the lemma statement.
However, notice that by definition, $\Lambda_{i-1}\geq 2\Lambda_{i}$, so we shall work instead with a reformulation of \Cref{lem:matching} as:
\begin{equation*}
  \Lambda_{i-1} \geq \sum_{j={i+1}}^\infty|M_j|2^{j-i} \enspace .
\end{equation*}

Again, we will show approximate strong optimality for the algorithms solution, so we obtain a constant factor approximation for all monotone symmetric norms.
\begin{lemma}
\label{lem:strongOptimalityDepots}
  For every $A\subseteq D$ we can find set $B\subseteq D$ with $|B|\leq |A|$ such that
  \begin{equation*}
    \sum_{o\in A} w(\tset_o) \leq 4 \sum_{o\in B}w(\that_o),
  \end{equation*}
  where $\{\tset_o\}_{o\in A}$ is the set of trees computed by the assignment algorithm for the depots in $A$.
\end{lemma}

\begin{proof}
  Fix some $A\subseteq D$, and for every $o\in A$ set $r_o$ to be the iteration of the algorithm in which~$\tset_o$ was last assigned a tree.
  Hence, $w(\tset_o) \leq 2\cdot 2^{r_o}$.
	
  Now we shall demonstrate how to find some comparatively-weighted trees among the~$\that_o$.
  We want to match every $\tset_o$ to one of the fragments of weight $2^{r_o-1}$ accounted for by~$\Lambda_{r_o-1}$.
  However, the fragments are not disjoint between the different $\Lambda_i$, i.e., a tree $\that_o$ of weight~$2^i$ will contribute an amount of $1$ to $\Lambda_i$, an amount of $2$ to $\Lambda_{i-1}$, etc., but we can not use that contribution repeatedly.
  To avoid overcounting, we thus execute a little trick:
  Assume, without loss of generality, that $\Lambda_{i-1} = \sum_{j={i+1}}^\infty|M_j|2^{j-i}$.
  To justify this assumption, note that we can decrease the weights of the $\that_o$ until the equalities hold.
  This decrease only worsens the constant we can get for this lemma.
  Observe now that 
  \begin{equation*}
    \Lambda_{i-1} -2\cdot \Lambda_i =  \sum_{j={i}}^\infty|M_j|2^{j-i}- \sum_{j={i+1}}^\infty|M_j|2^{j-i}  = |M_i| \enspace .
  \end{equation*}
  Therefore, we can partition our $\that_o$ as follows: 
  For the final iteration $i^*$ of the algorithm we can extract $\Lambda_{i^*-1}=|M_{i^*}|$ many fragments of weight $2^{i^*-1}$.
  For $i^*-1$ we can then partition each of these into two fragments of weight $2^{i^*-2}$.
  However, we can extract an additional $\Lambda_{i^*-2} -2\cdot \Lambda_{i^*-1} = |M_{i^*-1}|$ fragments, disjoint from the ones we already have.
  We continue this argument downwards.
  Thereby, we partition the $\that_o$ into $M_i$ disjoint fragments of weight~$2^{i-1}$.
	
  Now by construction, the number of $o\in D$ for which $r_o = i$ can be at most $|M_i|$. 
  Thus, we assign each $\tset_o$ with $o\in A$ uniquely to one of the fragments of weight $2^{r_o-1}$ we are guaranteed.
  We then take $B$ to be the set of trees $\that_o$ containing one of the chosen fragments.
  This clearly gives $|B|\leq |A|$, but it also guarantees:
  \begin{equation*}
    \sum_{o\in B} w(\that_o) \geq \sum_{o\in A} 2^{r_o-1} \geq \frac{1}{4} \sum_{o\in A}w(\tset_o) \enspace .\qedhere
  \end{equation*}
\end{proof}

The proof of \Cref{thm:TCdepots} now is a direct consequence of \Cref{lem:strongOptimalityDepots}, where the computation of the exact approximation factor which is obtained by this algorithm is left open. 
The present analysis suggests some very large number, but to maintain some presentability it is also highly inefficient.
Many of the reductions made could in principle be omitted at the cost of a much more detailed analysis.

\clearpage
\pagebreak

\section{Improved Analysis for the No-Depot Setting}
\label{apx:sec:noDepotImproved}
Speaking of constants of the approximation factor, we \emph{will} give an improved analysis for the approximation factor attained by the modified algorithm of Even et al.~\cite{even2004min} with respect to our all-norm objective in the non-depot setting since we can in fact get a reasonably small constant for this setting.
Obtaining a smaller approximation factor involves centrally one observation, namely that the previous analysis pays a considerable portion of its approximation factor to the fact that there can be up to $\frac{k}{2}$ large components, requiring it to give up a factor of $2$ in \Cref{lem:totalWeight}. 
However, notice that this is only the case if all large components have a spanning tree of weight less than $4R$.
This can be avoided by pretending that such components are in fact small.
Notice that we only use that all trees produced by the algorithm have weight at most $4R$, so this does not degrade the analysis, but it will allow us to assume that there are at most $\frac{k}{3}$ large components. 
Plugging this in, as well as rearranging the arithmetic a bit will yield a much smaller factor of approximation.

\begin{theorem}
\label{thm:14-apx}
  There is an algorithm that, given any edge-weighted graph $G$ and integer $k\in \mathbb{N}$, in polynomial time computes a set $\{T_1,\dots,T_t\}$ of at most $k$ trees  covering all nodes of $G$ such that, for any $p\in [1,\infty)$ and tree cover $\{\that_i,\dots,\that_k\}$ we have 
  \begin{equation*}
    \left(\sum w(T_i)^p\right)^{1/p} \leq 14\left(\sum w(\that_i)^p\right)^{1/p} \enspace .
  \end{equation*}
\end{theorem}
\begin{proof}
  Suppose that the algorithm decomposes the graph into $s$ small components of weight at most~$4R$, and $\ell$ large components of weight larger than $4R$, and fix some tree cover $\{\that_1,\dots,\that_k\}$.
	
  Denote by $s_1$ the number of small components for which some $\that_i$ is a spanning tree, by $s_2$ the number of small components which intersect multiple $\that_i$, and by $s_3$ the number of other small components, which are in particular incident to some edge $e$ of a $\that_i$ with $d(e) \geq R$.
	
  We can make the following observations, analogously to the old analysis:
  \begin{enumerate}
    \item $\sum w(\that_i) \geq \frac{s_3}{2}R + \sum_{\that_i \in T^=}w(\that_i)$,
    \item $\sum w(\that_i) \geq (k-s-\ell)2R  -(k-s-\ell-s_2)R + \sum_{\that_i \in T^=}w(\that_i)$,
    \item and $\ell \leq \frac{k-s}{3}$,
  \end{enumerate}
  where in the final point we now get a factor $1/3$ rather than $1/2$.
  We will now also apply the third inequality immediately, allowing us to get:
  \begin{align*}
    \sum w(\that_i) &\geq (k-s-\ell)2R  -(k-s-\ell-s_2)R + \sum_{\that_i \in T^=}w(\that_i)\\
    &= (k-s_1-s_3-\ell)R + \sum_{\that_i \in T^=}w(\that_i)\\
    &\geq (\frac{2}{3}k + \frac{1}{3}s - (s_1+s_3))R + \sum_{\that_i \in T^=}w(\that_i)\\
    &\geq \frac{2}{3}(k-s_1-s_3)R + \sum_{\that_i \in T^=}w(\that_i) \enspace .
  \end{align*}
  To remove the $s_3$, we take a convex combination of this inequality with\linebreak $\sum w(\that_i) \geq \frac{s_3}{2}R + \sum_{\that_i \in T^=}w(\that_i)$ and obtain:
  \begin{align*}
    \sum w(\that_i) &\geq \frac{3}{7}\cdot\frac{2}{3}(k-s_1-s_3)R + \frac{4}{7}\cdot\frac{s_3}{2}R + \sum_{\that_i \in T^=}w(\that_i)\\
    &=\frac{2}{7}(k-s_1)R + \sum_{\that_i \in T^=}w(\that_i) \enspace .
  \end{align*}
  This gives 
  \begin{align*}
    \left[\sum w(T_i)^p\right]^{1/p} &\leq \left[(k-s_1)(4R)^p + \sum_{\that_i \in T^=}w(\that_i)^p\right]^{1/p}\\
    &\leq 14 \left[(k-s_1)\left(\frac{2}{7}R\right)^p + \sum_{\that_i \in T^=}w(\that_i)^p\right]^{1/p}\\
    &\leq 14 \left[\sum w(\that_i)^p\right]^{1/p}  \enspace .\hfill\tag*{\qedhere}
  \end{align*}
\end{proof}
Notice again that the result is actually a lower bound on the average tree weight in any tree cover, excepting those trees that already agree with the solution computed by the algorithm. 
So this analysis does indeed give $c$-approximate strong optimality.

We can obtain an even stronger bound by exploiting that the number $s_3$ of small components witnessing some edge $e$ with $d(e) \geq R$ must be small:
\begin{lemma}
\label{lem:boundLargeEdges}
  Using the same notation as in the proof of \Cref{thm:14-apx}, if $s_3 \geq \frac{1}{3}(k-s_1)$, the algorithm computes a solution which is an $8$-approximation.
\end{lemma}
\begin{proof}
  Let $\that_i$ be some tree from an optimal solution to $(\ell_p,k)$ tree cover that contains fully covers a small component, as well as some edge $e$ with $d(e) \geq R$.
  Then take $T_1, \dots, T_j$ to be the spanning trees of those small components fully covered by $\that_i$. 
  We can notice that $\that_i$ contains at least $j/2$ edges of weight greater than $R$.
  Hence, we obtain $w(\that_i)^p \geq \sum_j w(T_j)^p + (j/2) R^p$.
	
  That is, in the algorithm all trees in the components counted by $s_3$ are being paid for by the optimal solution, and the optimal solution pays an additional $R^p$ for half of them, which we can use to cover the cost of some of the algorithms large trees.
	
  So let $T_{3}$ be the set of minimum spanning trees of the components counted by $s_3$.
  We obtain the following bounds:
  \begin{align}
    \sum_i w(\that_i)^p &\geq \sum_{T\in T^=} w(T)^p+ \sum_{T\in T_3} w(T)^p+ \frac{s_3}{2}R^p; \label{eq:boundOPTp} \\
    \sum_i w(T_i)^p &\leq \sum_{T\in T^=} w(T)^p+ \sum_{T\in T_3} w(T)^p+ \frac{k-s_1-s_3}{2}(4R)^p \enspace . \label{eq:Alg}
  \end{align}
  Here, \Cref{eq:Alg} follows from the fact that the trees of the algorithm have average weight $\leq 2R$ because of \Cref{alg:item:exactlyK}, and so the worst possible arrangement with respect to any $\ell_p$ objective is to have~$k/2$ trees of weight $4R$. 
  To obtain \Cref{eq:Alg}, we then additionally exclude the trees counted by~$s_1$ and~$s_3$ from that computation.
	
  Together, \Cref{eq:Alg,eq:boundOPTp} yield
  \begin{align*}
    \frac{\sum_i w(T_i)^p}{\sum_i w(\that_i)^p} \leq \frac{k-s_1-s_3}{s_3} 4^p \leq 2\cdot 4^p  \leq  8^p \enspace .
  \end{align*}
  This proves the lemma.
	
  Interestingly, note that for larger $p$ we can get much better bounds here.
  Namely, for $s_3 \geq \lambda (k-s_1)$ we would get
  \begin{equation*}
    \left(\frac{\sum_i w(T_i)^p}{\sum_i w(\that_i)^p}\right)^{1/p} \leq \left(\frac{k-s_1-s_3}{s_3}\right)^{1/p} 4 \leq \left(\frac{1-\lambda}{\lambda}\right)^{1/p}\cdot 4 \enspace .
  \end{equation*}
  Here, we can choose $\lambda$ quite small as $p$ grows, for example for $p = 2$ we could already take $\lambda = \frac{1}{5}$, obtaining better constants for specific, large values of $p$.
\end{proof} 
\begin{theorem}
\label{thm:9-apx}
  There is an algorithm that, given any edge-weighted graph $G$ and integer $k\in \mathbb{N}$, in polynomial time computes a set $\{T_1,\dots,T_t\}$ of at most $k$ trees  covering all nodes of $G$ such that, for any $p\in [1,\infty)$ and tree cover $\that_i,\dots,\that_k$ we have 
  \begin{equation*}
    \left(\sum w(T_i)^p\right)^{1/p} \leq 9\left(\sum w(\that_i)^p\right)^{1/p} \enspace .
  \end{equation*}
\end{theorem}
\begin{proof}
  We follow the proof of \Cref{thm:14-apx} to obtain
  \begin{align*}
    \sum w(\that_i) \geq \frac{2}{3}(k-s_1-s_3)R + \sum_{T_i \in T^=}w(T_i) \enspace .
  \end{align*}
  Then by \Cref{lem:boundLargeEdges} we can assume $s_3 \leq (k-s_1)/3$, otherwise we already have the statement of the theorem.
  So we obtain
  \begin{align*}
    \sum w(\that_i) \geq \frac{4}{9}(k-s_1)R + \sum_{T_i \in T^=}w(T_i),
  \end{align*}
  and thus
  \begin{align*}
    \sum w(\that_i)^p \geq (k-s_1)(\frac{4}{9}R)^p + \sum_{T_i \in T^=}w(T_i)^p \enspace .
  \end{align*}
  Together with 
  \begin{equation*}
    \sum w(T_i)^p \leq (k-s_1)(4R)^p + \sum_{T_i \in T^=}w(T_i)^p,
  \end{equation*}
  this implies the theorem.
\end{proof}

\clearpage
\pagebreak

\section{Computational Hardness}
\label{sec:Hardness}
\subsection{Weak \texorpdfstring{$\mathsf{NP}$}{}-hardness for\texorpdfstring{ $\ell_p$-{\sc Tree Cover with Depots}}{}}
We will show $\mathsf{NP}$-hardness for $\ell_p$-{\sc Tree Cover with Depots} by reduction from {\sc Partition}. 
In the {\sc Partition} problem, we are given a collection of integers $x_1,\dots,x_n$; the task is to compute a \emph{solution} in form of a subcollection $X\subseteq \{x_1,\dots,x_n\}$ such that $2\sum_{x\in X} x = \sum_{i=1}^nx_i$, or to decide that no solution exists.
It is known that {\sc Partition} is (weakly) \NP-complete \cite{Karp1972}.

We transform an instance $\{x_1,\dots,x_n\}$ of {\sc Partition} to an equivalent instance of $\ell_p$-{\sc Tree Cover with Depots} by constructing a graph $G$ with node set $V(G)=\{d_1,d_2\} \cup \{v_1,\dots,v_n\}$ and edge set $E(G) = \{d_1,d_2\} \times \{v_1,\dots,v_n\}$, i.e., $G$ is the complete bipartite graph $K_{2,n}$.
We then set $D=\{d_1,d_2\}$ and $w(d_i, v_j) = x_j$. 

One can briefly observe that, for any given subset $X$ of the nodes, there is a minimum Steiner tree for $X$ that consists of a star centered on one of the depots, plus possibly one additional edge to connect the other depot in case $\{d_1,d_2\}\subseteq X$. 
In particular this means that in an optimal $\ell_p$-{\sc Tree Cover with Depots} every tree will contain exactly one depot, so the optimal solutions correspond to partitions $X_1\dot \cup X_2$ of $\{v_1,\dots,v_n\}$ and have solution value $[(\sum_{v_i\in X_1}x_i)^p + (\sum_{v_i\in X_2}x_i)^p]^{1/p}$.
Since the $p$-norms are strongly convex for $p>1$, the instance $(G,D,w,p)$ then has optimal solution value $[2(\frac{1}{2}\sum_{i=1}^n x_i)^p]^{1/p}$ if and only if the original instance of {\sc Partition} admits a solution.

\subsection{\texorpdfstring{$\mathsf{W}[1]$}{}-hardness for \texorpdfstring{$\ell_p$}{}-{\sc Tree Cover with Depots} parameterized on~\texorpdfstring{$|D|$}{}} 
The previous construction can be expanded by allowing more than $2$ depots, say $k>2$.
The resulting instance will again have optimal solution value $[k(\frac{1}{k}\sum_{i=1}^n x_i)^p]^{1/p}$ if and only if there exists a partition of $X$ into $k$ bins $X_1,\dots,X_k$ such that the items in each bucket sum up to the same weight. 
Using this construction we can then solve {\sc Unary Bin Packing} with $k$ bins which was shown by Jansen et al. to be $\mathsf{W}[1]$-hard when parameterized by $k$ \cite{Jansen2013UnaryBinPacking}. 
Note here that in the case of unary bin packing, we may assume, without loss of generality, that all bins must be completely filled, by providing additional items of unit size.

\subsection{\texorpdfstring{$\mathsf{APX}$}{}-hardness for \texorpdfstring{$\ell_p$-{\sc Tree Cover with Depots}}{}}\label{ssec:hardnessLPTC}
Although the earlier constructions are relatively straightforward, determining a lower bound for the approximability of $\ell_p$-{\sc Tree Cover with Depots} requires  more work.
We proceed by reducing from a variant of maximum satisfiablitity of linear equations over $\mathbb{F}_2$ of arity~$3$, MAX-E3LIN2 for short.
In this problem, we are given a set of variables $x_1,\dots,x_n$, as well as a set of $m$ ternary equations $x\oplus y\oplus z = b$ with $b\in \{0,1\}$.
Here, $\oplus$ signifies addition modulo~$2$.
The goal is to compute an assignment of $\{0,1\}$ values to the variables which satisfies a maximum number of equations.

In addressing this problem, H\aa stad~\cite{Hastad99Inapx}, in his seminal work, demonstrated that for any $\varepsilon > 0$, it is $\mathsf{NP}$-hard to distinguish instances where $(1-\varepsilon)m$ equations can be satisfied from those where at most $(\frac{1}{2} + \varepsilon)m$ can be satisfied. 
For our purposes we need to also bound the variable degree.
A modification of  MAX-E3LIN2 achieving bounded variable degree was used by Karpinski et al.~\cite{Karpinski2015ApxTSP} to show improved lower bounds for the approximability of the {\sc Traveling Salesperson Problem}.
They proved that MAX-E3LIN2 remains hard to approximate even if every variable occurs in exactly three equations. 
In return they need to allow some equations of arity $2$, as well as a degradation in the approximability lower bound.
Specifically, they prove
\pagebreak

\begin{theorem}[\cite{Karpinski2015ApxTSP}, Theorem 3] \label{thm:maxe3lin2}
  For every $\varepsilon>0$ there exists a collection of systems of linear equations with $31m$ equations and $21m$ variables over $\mathbb{F}_2$ such that:
  \begin{itemize}
    \item each variable occurs in exactly three equations;
    \item $30m$ of the equations contain $2$ variables;
    \item $m$ of the equations are of the form $x\oplus y \oplus z = 0$;
    \item and it is $\mathsf{NP}$-hard to decide whether some assignment to the variables satisfies $(31-\varepsilon)m$ equations, or if every assignment satisfies at most $(30.5 + \varepsilon)m$ equations.
  \end{itemize}
\end{theorem}

This regularisation is achieved by taking an instance of MAX-E3LIN2 and replacing every occurrence of a variable with a new, dedicated variable, and then linking these new variables together by additional equations to force them to be equal if they correspond to the same original variable.
This increases the size of the system, hence there are $31m$ instead of $m$ equations.
We will call this problem \maxsat{}, for $3$-regular linear equations of arity $2$ or~$3$ over the binary domain. 
Its instances are denoted by $(X,\mathcal{C})$, where $X$ is the set of variables and $\mathcal{C}$ the set of clauses.

\paragraph*{Note:} The hardness of approximation for \maxsat{} presented by Karpinski et al.~\cite{Karpinski2015ApxTSP} relies on reducing the variable degree of an instance of MAX-E3LIN2 by means of a certain type of amplifier graph. 
They show how to construct such amplifiers probabilistically in polynomial time, but a deterministic construction is not known. 
Thus, the reduction only works probabilistically. 
If one desires deterministic guarantees, we may settle for a weaker statement by noting that for an instance with $m$ clauses their transformation requires at most one such graph for each $i \in 1, \dots m$, and the graph for each $i$ can be chosen independently of the instance and has size linear in $i$.
Thus their results also hold in a deterministic setting if we allow for polynomial runtime \emph{and} a polynomial advice string, i.e., assuming $\mathsf{NP}\not \subseteq \mathsf{P}/\mathsf{poly}$ rather than $\PP \not = \NP$.\vspace{0.3cm}\\
For an instance of {\sc 3R\{2,3\}L2} we may construct an instance of $\ell_p$-{\sc Tree Cover with Depots} by first introducing some gadgets for the variables and clauses:
\begin{itemize}
  \item For every variable $x$, introduce three nodes $\hat{x}, x_0,$ and $x_1$, as well as edges ${\hat{x}, x_i}$ for $i=0,1$ with weight $3$.
  Add the $x_i$'s as depots.
  \item For every ternary clause $C$, introduce nodes $\hat{C}, C_{000}, C_{110}, C_{101},$ and $C_{011}$ with edges $\{\hat{C}, C_i\}$ of weight $3$.
    Add the $C_i$'s as depots.
  \item For every binary clause $C = x\oplus y = 0$, introduce nodes $\hat{C}, C_{00},$ and $C_{11}$ with edges $\{\hat{C}, C_i\}$ of weight $2$.
    Add the $C_i$'s as depots. Further add nodes $\hat{C_{00}}$ and $\hat{C_{11}}$ where $\hat{C_i}$ is connected only to $C_i$ by an edge of weight $1$.
  \item For every binary clause $C = x\oplus y = 1$, introduce nodes $\hat{C}, C_{01},$ and $C_{10}$ with edges $\{\hat{C}, C_i\}$ of weight $2$.
    Add the $C_i$'s as depots.
    Further add nodes $\hat{C_{01}}$ and $\hat{C_{10}}$ where $\hat{C_i}$ is connected only to $C_i$ by an edge of weight $1$.
\end{itemize}

We connect the gadgets together as follows:
\begin{itemize}
  \item For every clause $C = x\oplus y \oplus z = 0$ we connect $C_{b_1b_2b_3}$ to $x_{b_1}, y_{b_2},$ and $z_{b_3}$ with a path of length $2$ where both edges have weight $1$.
  \item For every clause $C = x\oplus y = b$ we connect $C_{b_1b_2}$ to $x_{b_1}, y_{b_2}$  with a path of length $2$ where both edges have weight $1$.
\end{itemize}

An illustration of this construction is given in \Cref{fig:ApxGadgets}. 
Notice that there are $106m$ depots.

The goal of this construction is to represent the original system of linear equations in the following sense: For every variable $x$, the node $\hat{x}$ should be assigned to either $x_0$ or $x_1$. 
As this is relatively expensive, we would like to not assign any further nodes to that depot. 
The other depot, however, is still free, and should correspond to the variable assignment we want to choose.

Similarly, for every clause $C = x \oplus y \oplus z  = 0$ we need to assign $\hat{C}$ to one of the depots in the clause gadget. 
These correspond directly to the way in which we would like the clause to be satisfied. 
That is, assigning $\hat{C}$ to $C_{011}$ means that we have set $x = 0$, $y = 1$, and $z=1$ to satisfy the clause.

Now we may note that, if those choices are consistent between a clause and its incident variables, we will be able to also assign the nodes on the paths that we introduced between the gadgets. 
Specifically, every such node will be adjacent to a depot that has not yet been assigned any node.
We try to assign all of them to the neighbouring $C_i$ if it is empty.
The remainder we then assign to the adjacent $x_i$.
Since every variable appears in most 3 clauses, this will leave every tree with weight at most $3$, if we had a satisfying assignment to begin with.
In fact every tree will have weight exactly~$3$ in this setting, since the variable degrees are exactly~$3$.

Conversely, every clause that is not satisfied will cause one tree somewhere to have weight at least $3$. For every $\ell_p$ norm other than $\ell_1$ this imbalance in the tree weights will then be measureable in the value of an optimal solution.

\begin{figure}
  \centering
  \begin{subfigure}{.45\textwidth}
    \centering
    \includegraphics[width=0.9\textwidth, page = 1]{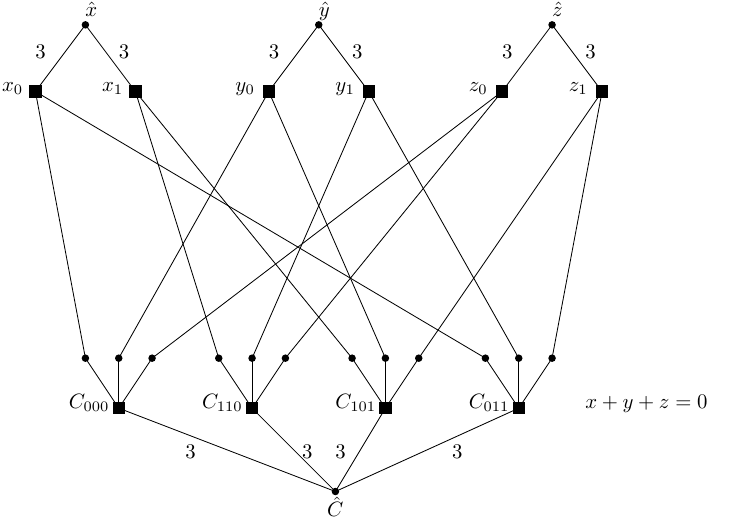}
    \caption{The gadget for a ternary clause $x \oplus y \oplus z = 0$, along with the gadgets for $x,y,z$ and the respective connections.}
  \end{subfigure}\hspace{0.09\textwidth}
  \begin{subfigure}{.45\textwidth}
    \centering
    \includegraphics[width=0.9\textwidth, page = 2, trim = -3cm 0 3cm 0]{APXhardness.pdf}
		\caption{The gadget for a binary clause $x \oplus y = 0$, along with the gadgets for $x,y$ and the respective connections.\phantom{sssssssssssssssssssssss}}
  \end{subfigure}
  \caption{Illustration of clause and variable gadgets for $3$ and $2$-ary constraints}
\label{fig:ApxGadgets}
\end{figure}

\subsubsection{Completeness}
We begin by describing how we can transform solutions to the $(31-\varepsilon)m$-satisfiable instances of {\sc 3R\{2,3\}L2} to good solutions to the constructed instances of \lptc{}.
\begin{lemma}
\label{lem:completeness}
  Let $(X,\mathcal{C})$ be an instance of \maxsat{} in which $(31-\varepsilon)m$ clauses can be simultaneously satisfied.
  Let $(G,D,w)$ be the instance of \lptc{} obtained from~$(X,\mathcal{C})$.
  Then
  \begin{equation*}
    OPT_p \leq  [(106-2\varepsilon)m 3^p + \varepsilon (2^p + 4^p)]^{1/p} \enspace .
  \end{equation*}
\end{lemma}

\begin{proof}
  Let $(G,D,w)$ be the instance obtained from transforming $(X,\mathcal{C})$ and take $f:X\to \{0,1\}$ to be an assignment to the variables satisfying $(31-\varepsilon)m|\mathcal{C}|$ clauses.
	
  To construct a tree cover of $G$, we do the following:
  \begin{enumerate}
    \item For every $x\in X$, assign $\hat{x}$ to $x_{1-f(x)}$.
    \item For every satisfied clause $C = x\oplus y\oplus z = 0$, assign $\hat{C}$ to $C_{f(x)f(y)f(z)}$.
    \item For every satisfied clause $C = x\oplus y = b$ with $b\in \{0,1\}$, assign $\hat{C}$ to $C_{f(x)f(y)}$, as well as $\hat{C}_{b_xb_y}$ to $C_{b_xb_y}$ for any $b_x,b_y\in \{0,1\}$.
    \item For every unsatisfied clause $C = x\oplus y\oplus z = 0$, assign $\hat{C}$ to $C_{f(x)f(y)1-f(z)}$.
    \item For every unsatisfied clause $C = x\oplus y = b$ with $b\in \{0,1\}$, assign $\hat{C}$ to $C_{f(x)1-f(y)}$, as well as $\hat{C}_{b_xb_y}$ to $C_{b_xb_y}$ for any $b_x,b_y\in \{0,1\}$.
    \item Finally, assign every of yet unassigned node by the following rules: 
      \begin{enumerate}
        \item If it has a neighbouring depot $C_I$ for some $C\in \mathcal{C}, I\in \{0,1\}^*$ such that $C_I$ has not been assigned $\hat{C}$, assign the node to this depot. 
        \item Otherwise, if it has a neighbour $x_f(x)$ for some $x\in X$, then assign it to this depot.
        \item Finally, if it is not assigned by the previous rules, it is incident to either $C_{f(x)f(y)1-f(z)}$ or $C_{f(x)1-f(y)}$ for some clause $C$, then assign it to this depot.
      \end{enumerate}
    \end{enumerate}
    We observe that---because every variable appears in at most three clauses---the application of this rule will cause every depot to be assigned a total weight of at most $3$, except those corresponding to the $C_I$ that was assigned $\hat{C}$ for some unsatisfied clause $C$; these depots are assigned a total weight of~$4$, consisting of $\hat{C}$, possibly $\hat{C}_I$, and one additional node. 
    However, there are only $\varepsilon m$ such depots.
	
    Noting that the total weight of all trees is $3\cdot106\cdot m$, we obtain that the value of this solution with respect to the $\ell_p$-norm is 
    \begin{equation*}
      [(106-2\varepsilon)m 3^p + \varepsilon (2^p + 4^p)]^{1/p} \enspace.\tag*{\qedhere}
  \end{equation*}
\end{proof}

\begin{figure}
  \centering
  \begin{subfigure}{.45\textwidth}
    \centering
    \includegraphics[width=0.8\linewidth, page = 3]{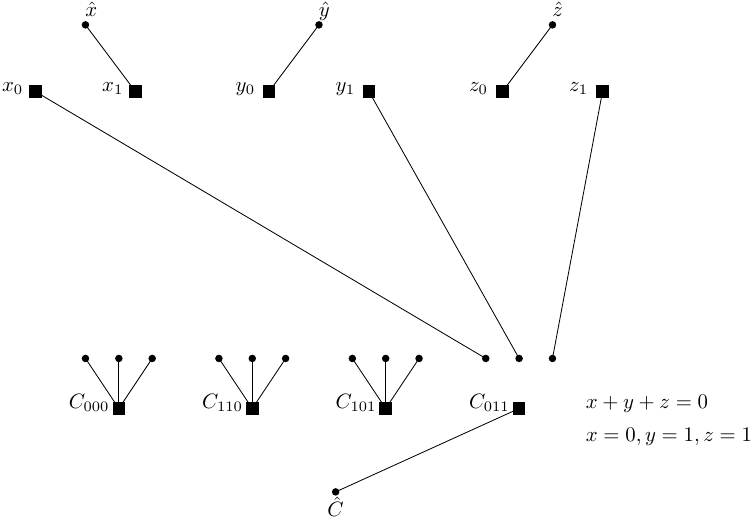}
    \caption{The gadget for a satisfied ternary clause under variable assignment $x=0$, $y=1$, $z=1$ with the selected solution. All trees have weight~$3$, or belong to the depot corresponding to a variables assigned value.}
  \end{subfigure}\hspace{0.09\textwidth}
  \begin{subfigure}{.45\textwidth}
    \centering
    \includegraphics[width=0.8\linewidth, page = 4]{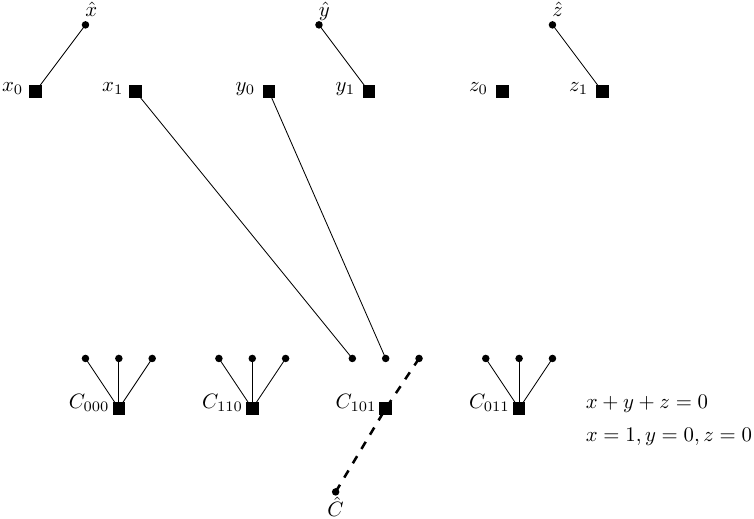}
    \caption{The gadget for an unsatisfied ternary clause under variable assignment $x=1$, $y=0$, $z=0$ with the selected solution. The tree with the dashed edges will have weight $4$, witnessing the fact that the clause was unsatisfied.}
  \end{subfigure}
  \caption{Illustrations of the appearance of Tree Cover solutions on satisfied and unsatisfied clauses.}
  \label{fig:ApxGadgetsSatisfaction}
\end{figure}

\subsubsection{Soundness}
We will conduct the soundness analysis in two stages, for sake of clarity.
First, we will look at solutions to the \lptc{} which are \emph{honest} in the sense that they assign every node to a neighbouring depot.
These rather closely correspond to solutions of the original instance of \maxsat{}, so we will be able to essentially show that if the instance was at most $(30.5+\varepsilon)m$-satisfiable, there are at least $(\frac{1}{2}-\varepsilon)m$ trees of weight at least $4$ in any honest solution.

In a second step we then show that one can not gain much by using a \emph{dishonest} solution instead.
Assigning a node to a depot that is not neighbouring incurs some cost with respect to~$\ell_1$, which will drive up the costs in all other $\ell_p$ norms as well.

\begin{lemma}
\label{lem:honestSolutions}
  Any honest solution $\{T_d\}_{d\in D}$ to some instance $(G,D,w)$ of \lptc{} derived from a $(31-c)m$-satisfiable instance $(X,\mathcal{C})$ of \maxsat{} will have 
  \begin{equation*}
    \left[\sum_{d\in D} w(T_d)^p\right]^{1/p} \geq \left[(106-2c)m 3^p + cm(2^p+4^p)\right]^{1/p} \enspace .
  \end{equation*}
\end{lemma}

\begin{proof}
  Suppose we are given an honest solution to the \lptc{} instance constructed from an instance of \maxsat{}. 
  We obtain an assignment $f$ to the variables by setting $f(x) = 1$ if $\hat{x}$ is assigned to $x_0$, and vice versa. 
  Now take any clause $C$ not satisfied by~$f$ and let $C_I$ be the depot which was assigned $\hat{C}$. 
  Then there exists at least one node in $N(C_I) \setminus \{\hat{C}\}$ such that it is adjacent to two depots, both of which are already have weight at least $3$ assigned to them. 
  We call the set of all such nodes $V_{\textnormal{excess}}$.
	
  These nodes must be assigned to one of their neighbouring depots, causing it to have weight at least $4$.
  We formally detect these overloaded depots by estimating the following:
  \begin{equation*}
    \sum_{d\in D} \max\{w(T_d)-3, 0\} \geq |V_{\textnormal{excess}}| \geq cm. 
  \end{equation*}
  Noting again that we have $\sum_{d\in |D|} w(T_d) \geq 318m$, and that a lower bound on the $\ell_p$ value can be obtained by distributing the total weight as evenly as possible, under the constraint that \(\sum_{d\in D} \max\{w(T_d)-3, 0\}\geq cm\), we then see that
  \begin{equation*}
    \left[\sum_{d\in D} w(T_d)^p\right]^{1/p} \geq \left[(106-2c)m 3^p + cm(2^p+4^p)\right]^{1/p} \enspace . \tag*{\qedhere}
  \end{equation*}
\end{proof}

\begin{lemma}
\label{lem:dishonestSolutions}
  Any solution $\{T_d\}_{d\in D}$ to an instance $(G,D,w)$ of \lptc{} derived from a $(31-c)m$-satisfiable instance $(X,\mathcal{C})$ of \maxsat{} will have
  \begin{equation*}
    \left[\sum_{d\in D} w(T_d)^p\right]^{1/p} \geq \left[(106-\frac{c}{2}) 3^p + \frac{c}{4}(2^p+4^p)\right]^{1/p} \enspace .
  \end{equation*}
\end{lemma}

\begin{proof}
  We now also have to consider the dishonest solutions, i.e., solutions where some nodes were not assigned to a neighbouring depot. 
  For each such node there is some increase in the~$\ell_1$ value of the solution, so there can not be too many of them. 
  But then most of the instance looks like it is honest, since for every variable on which the solution cheats it can satisfy at most~$3$ additional clauses above what an optimal solution to \maxsat{} could.
	
  So let $V_{\textnormal{dishonest}} \subseteq V(G)$ be the set of nodes not assigned to a neighbouring depot. 
  Then we see that $\sum_{d\in D} w(T_d) \geq 318m + |V_{\textnormal{dishonest}}|$ and thus $\sum_{d\in D} \max\{w(T_d)-3, 0\} \geq |V_{\textnormal{dishonest}}|$.
	
  Now suppose we take the set of variables $X_{\textnormal{dishonest}} \subseteq X$ whose corresponding gadget contains a node from $V_{\textnormal{dishonest}}$, and similarly $\mathcal{C}_{\textnormal{dishonest}}$ for the clauses.
  We now delete from $(X,\mathcal{C})$ all variables in $X_{\textnormal{dishonest}}$, all clauses in $\mathcal{C}_{\textnormal{dishonest}}$, and finally all clauses containing a variable from $X_{\textnormal{dishonest}}$. 
  Call the resulting instance of \maxsat{} $(X',\mathcal{C}')$. 
  We notice that~$\mathcal{C}'$ still contains at least $31m-3|V_{\textnormal{dishonest}}|$ clauses. 
  Also, any assignment to $X'$ satisfying a clause $C\in \mathcal{C}'$ will also satisfy that clause in $(X,\mathcal{C})$, irrespective of how we extend it to $X \setminus X'$. 
  Thus any assignment to $X'$ will leave at least $cm - 3|V_{\textnormal{dishonest}}|$ clauses unsatisfied.
  Thus we have an honest assignment leaving some clauses unsatisfied, so \Cref{lem:honestSolutions} applies.
  We see that $\sum_{d\in D} \max\{w(T_d)-3, 0\}\geq cm - 3|V_{\textnormal{dishonest}}|$, and thus
  \begin{equation*}
    \sum_{d\in D} \max\{w(T_d)-3, 0\}\geq \max\{cm - 3|V_{\textnormal{dishonest}}|, |V_{\textnormal{dishonest}}|\} \geq \frac{c}{4}m \enspace .
  \end{equation*}
  This gives the desired relation
  \begin{equation*}
     \left[\sum_{d\in D} w(T_d)^p\right]^{1/p} \geq \left[(106-\frac{c}{2})m 3^p + \frac{c}{4}m(2^p+4^p)\right]^{1/p} \enspace .\tag*{\qedhere}
  \end{equation*}
\end{proof}

\subsubsection{Hardness of Approximation}
\Cref{lem:completeness,lem:dishonestSolutions} jointly show that our reduction can, up to a factor of $4$, detect the number of clauses that are satisfiable in an instance of \maxsat{}.
Therefore, the gap shown by Karpinski et al.~\cite{Karpinski2015ApxTSP} now transfers also to our setting:
\begin{theorem}
  \label{thm:apxHardness}
  For every $p\in (1,\infty)$ there exists a constant $c$ such that \lptc{} {\sc with Depots} is \NP-hard to approximate within a factor $c$. 
  The \NP-hardness holds under randomized reductions.
\end{theorem}
\begin{proof}
  Fix some $p$ and $\varepsilon$ and take an instance of \maxsat{} from an instance family for which it is \NP-hard to distinguish whether the instance is $(31-\varepsilon)m$-satisfiable or at most $(30.5+\varepsilon)m$. 
  Karpinski et al.~\cite{Karpinski2015ApxTSP} show that such a family exists under randomized reductions.
  Now suppose, for sake of contradiction, we had a polynomial time $c$-approximation algorithm for \lptc{}. 
  Apply this algorithm to an instance generated by our reduction and denote by $ALG$ the value of the computed solution. 
  If the initial instance of \maxsat{} was $(31-\varepsilon)m$-satisfiable, then we obtain
  \begin{equation*}
    ALG \leq c\cdot OPT_p \leq c[(106-2\varepsilon)m 3^p + \varepsilon m (2^p + 4^p)]^{1/p} \enspace
  \end{equation*}
  from \Cref{lem:completeness}.
  Otherwise, we have
  \begin{equation*}
    ALG \geq OPT_p \geq [(106-\frac{1}{4} + \frac{\varepsilon}{2})m 3^p + (\frac{1}{8} - \frac{\varepsilon}{4}) m(2^p+4^p)]^{1/p} \enspace .
  \end{equation*}
  So unless \PP=\NP, we must have 
  \begin{equation*}
    c[(106-2\varepsilon)m 3^p + \varepsilon m (2^p + 4^p)]^{1/p} \geq [(106-\frac{1}{4} + \frac{\varepsilon}{2})m 3^p + (\frac{1}{8} - \frac{\varepsilon}{4}) m(2^p+4^p)]^{1/p} \enspace .
  \end{equation*}
  Brief computation yields
  \begin{align*}
    c &\geq \frac{[(106-\frac{1}{4} + \frac{\varepsilon}{2})m 3^p + (\frac{1}{8} - \frac{\varepsilon}{4}) m(2^p+4^p)]^{1/p}}{[(106-2\varepsilon)m 3^p + \varepsilon m (2^p + 4^p)]^{1/p}}\\
    &\geq \left[ \frac{(106-\frac{1}{4} + \frac{\varepsilon}{2}) 3^p + (\frac{1}{8} - \frac{\varepsilon}{4}) (2^p+4^p)}{(106-2\varepsilon) 3^p + \varepsilon  (2^p + 4^p)}\right]^{1/p} > 1,
  \end{align*}
  where the final inequality follows immediately from the relation $2^p + 4^p > 2\cdot 3^p$, which holds for all $p\in (1,\infty)$.
\end{proof}
Notice that this computation does not hold at $p=1$, since $2^1 + 4^1 = 2\cdot 3^1$.
Indeed $\ell_1$-{\sc Tree Cover} is polynomial-time solvable. 
Further, the computation requires $p$ to be a real number, so the above reduction does not yield hardness for $\ell_{\infty}$-{\sc Tree Cover}. 
This should be unsurprising, since by \Cref{lem:completeness} there is always a solution whose value with respect to~$\ell_\infty$ is $4$, unless the original instance of \maxsat{} was satisfiable (recall that satisfibility of a system of linear equations over~$\mathbb{F}_2$ can of course be checked in polynomial time.).

\subsubsection{\texorpdfstring{$\mathsf{APX}$}{}-Hardness for \texorpdfstring{\litc{}}{} {\sc with Depots}}
Obtaining some hardness of approximation for \litc{} {\sc with Depots} is nonetheless possible and not too complicated.
Such hardness is already known due to the work of Xu and Wen \cite{xu2010approximation}; however, we restate their result here to demonstrate a reduction that is similar to the one used to prove \Cref{thm:apxHardness}.
We reduce from $3$-{\sc Occurrence SAT}, i.e., {\sc SAT} where every variable occurs in exactly $3$ clauses, which is known to be \NP-hard for example due to Berman et al. \cite{Berman073OccSat}.
We denote instances of $3$-{\sc Occurrence SAT} as $(X,\mathcal{C})$ where $X$ is the set of variables and $\mathcal{C}$ the set of clauses.
For any such instance we construct an equivalent instance of \litc{} {\sc with Depots} in similar fashion to before:
\begin{enumerate}
  \item For each $x\in X$, introduce nodes $x_0,x_1,\hat{x}$ with edges $\{x_i,\hat{x}\}$ of weight $2$.
    Take the $x_i$ to be depots.
  \item For each clause $C = x^1\vee x^2 \vee \dots x^s \vee \overline{y^1}\vee \overline{y^2} \dots \vee \overline{y^t}$ introduce a node $C$ and edges $\{C,x^i_1\}$, $\{C, y^i_0\}$ with weight $1$.
\end{enumerate}
Any solution to this instance with value $2$ with respect to $\ell_\infty$ will correspond exactly to a satisfying assignment $f$ of the instance $(X,\mathcal{C})$ by taking $f(x) = i$ where $i$ is such that $\hat{x}$ was assigned to~$x_{1-i}$.
Note that we do not have to worry here about the solution ``cheating'', since every $\hat{x}$ must be assigned to one of the neighbouring $X_i$'s, otherwise the solution value is at least~$3$.
Every clause is then satisfied by the variable to whose depot it has been assigned.

Conversely, any satisfying assignment $f$ can be transformed into a solution of the \litc{} instance with value $2$. 
We assign $\hat{x}$~to $x_{1-f(x)}$ for each variable $x\in X$, and $C$ to~$x_{f(i)}$, where~$x$ is a variable satisfying the clause $C$ with respect to $f$.
This achieves value~$2$, since---without loss of generality---every variable satisfies at most two clauses (otherwise it occurs only negated/non-negated and can be removed).

Thus, we cannot distinguish instances of \litc{} {\sc with Depots} with solution value~$2$ from those of solution value $3$ in polynomial time unless \PP=\NP, so \litc{} {\sc with Depots} is \NP-hard to approximate to within a factor $\frac{3}{2}-\varepsilon$ for every $\varepsilon > 0$.
\end{document}